%% file: main.tex
\pdfoutput=1
\newif\ifdraft \draftfalse 
\drafttrue
\documentclass[11pt]{article}
\usepackage[numbers]{natbib}
\usepackage{algorithm}
\usepackage[noend]{algpseudocode}

\usepackage{macros}
\usepackage{xspace}
\usepackage{xcolor}
\definecolor{darkblue}{rgb}{0,0,.5}
\usepackage[colorlinks=true,allcolors=darkblue]{hyperref}
\usepackage{bbm}
\usepackage{tikz}
\usetikzlibrary{matrix}
\usepackage{pgfplots}
\usepackage{overpic}
\usepackage{makecell}
\usepgfplotslibrary{fillbetween}
\usetikzlibrary{patterns}
\usepackage{authblk}
\usepackage{enumerate}
\usepackage{amssymb}
\usepackage{amsbsy}
\usepackage{amsmath}
\usepackage{amsthm}
\usepackage{amsfonts}
\DeclareMathOperator{\polylog}{polylog}

\usepackage{latexsym}
\usepackage{graphicx}
\usepackage{color}
\usepackage{xifthen}

\usepackage{latexsym}

\usepackage{subcaption}

\newcommand{\reals}{\mathbb{R}} 
\newcommand{\R}{\mathbb{R}}

\newcommand{\N}{\mathbb{N}}

\newtheorem{theorem}{Theorem}
\newtheorem{lemma}{Lemma}[section]

\newtheorem{definition}{Definition}[section]

\newtheorem{assumption}{Assumption}

\usepackage{fullpage}
\begin{document}

\title{Differentially Private Histograms under Continual Observation: Streaming Selection into the Unknown}
\author{Adrian Rivera Cardoso and Ryan Rogers}
\affil{Data Science Applied Research, LinkedIn}

\maketitle 

\begin{abstract}
We generalize the continuous observation privacy setting from \citet{DworkNaPiRo10} and \citet{ChanShSo11} by allowing each event in a stream to be a subset of some (possibly unknown) universe of items. We design differentially private (DP) algorithms for histograms in several settings, including top-$k$ selection, with privacy loss that scales with $\polylog(T)$, where $T$ is the maximum length of the input stream.  We present a meta-algorithm that can use existing one-shot top-$k$ private algorithms as a subroutine to continuously release DP histograms from a stream. Further, we present more practical DP algorithms for two settings: 1) continuously releasing the top-$k$ counts from a histogram over a known domain when an event can consist of an arbitrary number of items, and 2) continuously releasing histograms over an unknown domain when an event has a limited number of items. 
\end{abstract}

\input{intro}

\input{prelim}

\input{baseMechanism}

\input{metaAlgo}

\input{unrestrictedSensitivity}

\input{revisitUnkGauss}

\input{unkDomain3}

\input{conclusion}

\section{Acknowledgments}
We would like the thank the following people for helpful comments throughout this research project:  Parvez Ahammad, David Durfee, Souvik Ghosh, Koray Mancuhan, and Diana Negoescu.

\clearpage

\bibliography{bib}
\bibliographystyle{abbrvnat}

\clearpage

\appendix

\input{appendixOldAlgos}

\end{document}

%% file: intro.tex
\section{Introduction \label{sect:intro}} 

Providing real-time statistics on streaming data is a common task in data analytics. For example, one may want to provide a running count on the number of people that have purchased a particular drug at a pharmacy.  This data can be very useful for tracking and identifying local epidemics in a given region. However, this particular data is very sensitive so privacy techniques should be applied to protect those who are purchasing medications.  Differential privacy (DP) has emerged as the go to method in industry to provide privacy for aggregate results. In this work, we study the problem of continually releasing aggregate counts over a stream of incoming data subject to DP.

Let $\cU$  be a set of items and $\omega_{1:T} = \omega_1, \cdots, \omega_T$ be a stream of $T$ events,\footnote{In fact, $T$ need not be the actual length of the stream, $T$ could be an upper bound. This assumption is common in the streaming DP literature, see \cite{ChanShSo11,SmithTh13}. }  e.g. pharmacy purchases, where $\omega_t \subseteq \cU$.\footnote{Our setting easily extends to each event consisting of items in $\cU$ and counts of each item from that event.  We can accommodate for this more general setting by scaling the noise by the maximum amount any item can change in an event, i.e. the $\ell_\infty$-sensitivity} Our goal is to release, at every time $t$, the counts of all items in the substream $\omega_{1:t}$, or the most frequent counts, subject to DP. This setting is referred to as the continual observation model of DP and originated in works from \citet{DworkNaPiRo10} and \citet{ChanShSo11} where it is assumed that $\cU$ is known and $\vert \omega_t \vert \leq 1$. In this paper we study settings where $\cU$ is either known (Known Domain) or unknown (Unknown Domain), and where a bound $\Delta_0$ on $|\omega_t|$ is known (Restricted $\ell_0$-sensitivity) or where it can be as large as $\vert \cU \vert = d$ (Unrestricted $\ell_0$-sensitivity).  In the unrestricted $\ell_0$-sensitivity setting, we only want to return the top-$k$ counts, rather than the full set of counts and have privacy loss increase with $k$ or $\sqrt{k}$.  Simply applying restricted $\ell_0$-sensitivity algorithms in the unrestricted $\ell_0$-sensitivity setting would require setting $\Delta_0=d$, so that privacy loss increases with $d$ or $\sqrt{d}$.

The guarantee of a DP algorithm is that the output distributions for two similar input streams will be similar. As is common in the continual observation DP literature, we do not restrict the number of events $\omega_t$ that a user can impact, thus we provide \emph{event level} privacy guarantees, as opposed to \emph{user level} privacy. 
We refer the reader to \cite{KiferSoRoThZh20} for an excellent overview on the \emph{granularity} of privacy, which describes user and event level privacy as well as models of privacy between these extremes. 
In any of the settings we consider, we could apply the corresponding one-shot DP algorithms presented in Table~\ref{table:tasks} on the data available at time $t$. However, releasing a total of $T$ answers would cause the total privacy loss to scale as $O(\sqrt{T})$ (using advanced composition privacy loss bounds). The goal of this work is to design algorithms for all settings in Table~\ref{table:tasks} and have the total privacy loss scale as $O(\polylog(T))$, or equivalently have the noise that we include for DP scale with $O(\polylog(T))$ for a constant privacy loss.

\begin{table*}[htbp]
    \centering\setcellgapes{4pt}\makegapedcells
    \begin{tabular}{ |c|c|c| } 
     \hline
     & Restricted $\ell_0$-sensitivity & Unrestricted $\ell_0$-sensitivity \\ 
     \hline
     \shortstack{Known \\ Domain} & $\knownGauss{}$ \cite{DworkKeMcMiNa06} & $\knownGumb{}$ \cite{McSherryTa07} \\ 
     \hline
    \shortstack{Unknown \\ Domain} & $\texttt{LimitDom}_\lap$ \cite{DurfeeRo19} / $\unkGauss{}$  & $\unkGumb{}$ \cite{DurfeeRo19} \\ 
     \hline
    \end{tabular}
    \caption{DP algorithms for various data analytics tasks in the one-shot analytics setting.\label{table:tasks}}
\end{table*}

Existing DP algorithms for the continual observation setting include the celebrated Binary Mechanism \cite{ChanShSo11}, which can be applied to the known domain and restricted $\ell_0$-sensitivity setting.  To our knowledge, we are the first to consider the other three quadrants of Table~\ref{table:tasks} in the continual observation setting.  In particular, we are the first to study continually releasing the item with the maximum count and its count at each round subject to DP, despite the one-shot DP algorithm being the classical Exponential Mechanism \cite{McSherryTa07}.  Other works have considered the problem of continually returning the top-$k$ \cite{ChanShSo11} and heavy hitters in a stream \cite{ChanLiShiXu12}, \cite{MirMuNiWr11}.  The main difference in our setting is that a single event consists of multiple distinct items, while earlier work has events with at most one item, which falls under the restricted $\ell_0$-sensitivity with known domain setting.  Our setting provides stronger levels of privacy because a single event in a stream can affect the count of multiple items at once.  In the pharmacy example, an event would be a purchase occurring and the items would be the drugs that were purchased, which need not be a single drug.  Note that \citet{MirMuNiWr11} considers a more restrictive privacy model, referred to as \emph{pan-privacy} from \citet{DworkNaPiRoYe10}, that includes security considerations so that privacy is preserved even if an adversary can access internal states of the algorithm. 

We point out that \citet{DworkNaPiRo10} provides a general transformation from one-shot algorithms to those with privacy guarantees under continual observation.  However, this general transformation requires the one-shot algorithm to return a scalar, which is then compared with the algorithm's outcomes at later rounds and only displays the new outcome if it is significantly different than the previous result, otherwise it will show the old result.  Our one-shot algorithms return a histogram of counts with labels that can differ in each round, so it is not clear what scalar function to assign to determine when a new outcome should be used.  We will use a similar idea to this general transformation in Section~\ref{sect:sparseGumb} when continually returning the top-$k$ from a stream of events and only updating results if there is a count that should be in the top-$k$ but is not at a current round.  Our approach allows for the privacy loss to increase with the number of times the top-$k$ should be updated, rather than when the counts from the previous round's top-$k$ need to be updated due to counts increasing but the top-$k$ remaining unchanged as would be the case by using the approach in \cite{DworkNaPiRo10} without returning labels.   

We also design algorithms that can be used in scalable and distributed real-time analytics platforms where low latency is crucial, so retrieving and passing the algorithm a substream $\omega_{1:t}$ at each time step $t$ is not feasible. Instead, algorithms in this setting only have access to the histogram at time $t$. An example of such platform is described in detail in  \cite{RogersSuPeDuLeKaSaAh20}. The Binary Mechanism can be implemented in this setting, since we only need access to the true counts over all items at each round $t$, rather than the full sequence of events, as long as the algorithm knows the length of the stream $t$ and the noise it has used in previous rounds, which can be replicated via seeding. For the unrestricted $\ell_0$-sensitivity with known domain setting, we design an algorithm that combines the Binary Mechanism, the Exponential Mechanism, and the Sparse Vector technique \cite{DworkNaReRoVa09} to continually release the top-$k$.  We also show that the more practical version can closely match the error from the less practical version with access to the full event stream.  In the case when each event consists of at most $\Delta_0$ items from an unknown set (restricted $\ell_0$-sensitivity with unknown domain), we develop an algorithm that can be viewed as a combination of $\unkGauss{}$ (a variant of $\texttt{LimitDom}_\lap$ \cite{DurfeeRo19} with an improved privacy guarantee) for one-shot analytics and the Binary Mechanism \cite{ChanShSo11}.  

We now summarize our contributions. First, we develop a general way to apply existing one shot DP top-$k$ algorithms for the continual observation setting. Second, we design more practical continual observation DP algorithms for the restricted $\ell_0$-sensitivity with unknown domain ($\unkBase{}$) and for the unrestricted $\ell_0$-sensitivity with known domain ($\sparseGumb{}$), with utility results for both. Third, we present a unified argument for analyzing both $\unkBase{}$ and $\unkGauss{}$ that improves on prior analysis of $\texttt{LimitDom}_\lap$ from \cite{DurfeeRo19}, which might be of independent interest.

%% file: prelim.tex
\section{Preliminaries} 

Since we will provide event level privacy guarantees we define neighboring histograms as follows. Two streams $\omega_{1:T}$ and $\omega_{1:T}'$ are neighboring if for some $t \in [T] := \{1,\cdots,T\}$, $\omega_t \neq \omega_{t}'$ where $\omega_t = \emptyset$ or  $\omega_t' = \emptyset$ but $\omega_{t'} = \omega_{t'}'$ for all $t' \in [T]$ such that $t\neq t'$. We will denote a histogram $\bbh = \{ (h^u, u) : u \in \cU, h^u \in \N \}$ to include counts and labels in $\cU$. Given stream $\omega_{1:t}$ we define its histogram over $\cU$ as $
\bbh(\omega_{1:t}; \cU) \defeq \left\{ \left( h_t^u \defeq \sum_{\ell = 1}^t \1{u \in \omega_\ell},u \right) : u \in \cU \right\}$.  We will refer to $\epsilon$ as the privacy loss parameter in the definition of DP.

\begin{definition}[\citet{DworkMcNiSm06}, \citet{DworkKeMcMiNa06}]
A randomized algorithm $M: \cX \to \cY$ that maps input set $\cX$ to some arbitrary outcome set $\cY$ is $(\epsilon,\delta)$-DP if for any neighboring datasets $x,x'$ and outcome sets $S \subseteq \cY$,
$
\Pr\left[ M(x) \in S\right] \leq e^\epsilon \Pr\left[ M(x') \in S\right] + \delta.
$
When $\delta = 0$, we typically say that $M$ is $\epsilon$-DP or \emph{pure} DP.  
\end{definition}

The analysis of our algorithms will typically use a variant of DP called \emph{zero-mean Concentrated DP} (zCDP) from \citet{BunSt16}, which provides tighter composition bounds than traditional DP analysis.
This variant of DP is based on the R\'{e}nyi divergence of order $\alpha >1$ between two distributions $P$ and $Q$ over the same domain, denoted as $D_\alpha(P || Q)$ where
\[
D_{\alpha}(P||Q) \defeq \frac{1}{\alpha - 1} \log \E_{z \sim P} \left[ \left( \frac{P(z)}{Q(z)} \right)^{\alpha - 1}\right].
\]

\begin{definition}[Zero-mean Concentrated Differential Privacy]
A randomized algorithm $M: \cX \to \cY$ is $\delta$-approximately $\rho$-zCDP if for any neighbors $x, x' \in \cX$, there exists events $E$ and $E'$, such that $\Pr[E], \Pr[E'] \geq 1-\delta$ and for every $\alpha > 1$ we have the following bound in terms of the R\'enyi divergence $D_\alpha(\cdot || \cdot)$ of order $\alpha$
\begin{align*}
&D_\alpha(M(x)|_E || M(x')|_{E'} ) \leq \alpha \rho\text{, and}\\
&D_\alpha(M(x')|_{E'} || M(x)|_{E} ) \leq \alpha \rho.
\end{align*}
where $M(x)|_E$ is the distribution of $M(x)$ conditioned on event $E$ and similarly for $M(x') |_{E'}$. If $\delta = 0$, then we say $M$ is $\rho$-zCDP.
\label{defn:zCDP}
\end{definition}

A useful property of zCDP is that composing multiple zCDP mechanisms results in another zCDP mechanism where the privacy parameters add up.

\begin{lemma}[\citet{BunSt16}]\label{lem:zCDP_Comp}
Let $M_1:\cX \to \cY$ be $\delta_1$-approximate $\rho_1$-zCDP and $M_2:\cX\times \cY \to \cY'$ be $\delta_2$-approximate $\rho_2$-zCDP in its first argument, i.e. $M_2(\cdot, y)$ is $\delta_2$-approximate $\rho_2$-zCDP for all $y \in \cY$.  Then the mechanism $M: \cX \to \cY'$ where $M(\cdot) = M_2(\cdot, M_1(\cdot))$ is $(\delta_1 + \delta_2)$-approximate $(\rho_1+\rho_2)$-zCDP.
\end{lemma}
We will state our privacy guarantees in terms of zCDP or DP.  We can then convert zCDP to DP and back with the following result.
\begin{lemma}[\citet{BunSt16}]
  If $M$ is $(\epsilon,\delta)$-DP then it is $\delta$-approximate $\epsilon^2/2$-zCDP. If $M$ is $\delta$-approximate $\rho$-zCDP then $M$ is also $(\epsilon(\rho,\delta'),\delta+\delta')$-DP where 
\begin{equation}
\epsilon(\rho,\delta) \defeq \rho + 2\sqrt{\rho \ln(1/\delta')}.
\label{eq:epsilon}
\end{equation}
\end{lemma}

%% file: baseMechanism.tex
\section{ Binary Mechanism: Restricted $\ell_0$-sensitivity with Known Domain Setting \label{sect:knownBase}} 
We first discuss the classical Binary Mechanism from \citet{ChanShSo11} that provides a running count $y_t := \sum_{\tau=1}^t \sigma_\tau$, from a bit steam $\sigma_{1:T}$ where $\sigma_t \in \{ 0,1\}$. The Binary Mechanism works by maintaining a binary tree and adding the $t$-th event from the stream into the $t$-th leaf. As this is done, one has to make sure the sum at each node is equal to the sum of its children. To compute the private count at $t$ it suffices to add the (noisy) sums corresponding to step $t$. We map the tree of partial sums into a \emph{partial sum table} $\bbp$ with entries $p_{i,j}$ for $i\in[\log_2(T)],j\in[T/2^{i-1}]$. The Binary Mechanism has multiple applications, including private matchings \cite{HsuHuRoRoWu14}, congestion games \cite{RogersRo14}, and private online learning \cite{SmithTh13}, \cite{CardosoCu19}.

Due to recent work comparing the overall privacy loss for Laplace noise and Gaussian noise from \cite{CesarRo20} and \cite{CanonneKaSt20}, we will use Gaussian noise, rather than Laplace noise in the original algorithm.  Further, we note that there is nothing special with using a binary representation, so we will keep the base $r$ arbitrary and optimize the base for the lowest overall variance subject to a given privacy level.  Considering arbitrary bases for the Binary Mechanism was also considered in \cite{QardajiYaLi13}, although they optimize for the mean squared error and we consider the worst error on any count.  We show that although the optimal base depends on knowing the stream length $T$ in advance, we show that there are several choices of the base that will improve over base 2 for large ranges of $T$.
To help ease notation, we write
\begin{equation}\label{eq:baseL}
L_r \defeq \lfloor \log_r(T) \rfloor + 1.
\end{equation}

Let $s_j(t; r) \in \{ 0,1, \cdots r-1\}$ be the $j$th digit in the representation of $t$ with base $r$, i.e. $t = \sum_{j = 0}^{\lfloor \log_r(t) \rfloor} s_j(t;r) r^j$. Keeping the base $r$ arbitrary, we now present the generalized version of the Binary Mechanism in Algorithm~\ref{algo:BaseMechanism}, which we refer to as $\BinMech{}$.
\begin{algorithm}[h!]
	\caption{$\BinMech$; Return a running count}
	\begin{algorithmic}
		\State  \textbf{Input:} Stream $\sigma_{1:T} = \sigma_1, \cdots, \sigma_T$, where $\sigma_t\in \{ 0,1\}$, noise level $\tau$, and base $r \in \{2, \cdots, T \}$.
		\State  \textbf{Output:} Noisy counts $\hat{y}_{1:T} = \hat{y}_1, \cdots, \hat{y}_T$,$\hat{y}_t\in\reals$ for all $t\in[T]$
		\State  Sample $\{Z_{i,j} : i \in [L_r], \text{ and } j \in [T] \} \stackrel{i.i.d.}{\sim} \Normal{0}{ L_r \tau^2}$
		\For{$i \in [L_r]$} \algorithmiccomment{Populate the partial sum table}
			\State  $\texttt{START} = 1$
			\For{$j \in [\lfloor T/r^{i-1} \rfloor]$}
				\State  $\texttt{END} = \texttt{START}  + r^{i-1} - 1$
				\State  $p_{i,j} = \sum_{\ell = \texttt{START}}^{\texttt{END}} \sigma_\ell $
				\State  $\texttt{START}= \texttt{END} + 1$.
			\EndFor
		\EndFor
		\For{$t = 1, \cdots, T$}
			\State  Write $t = \sum_{j = 0}^{\lfloor \log_r(t)\rfloor}s_j(t;r) r^j$	
			\State  Let $t' = t$
			\State  $\hat{y}_t = 0$
			 \While{$ t' > 0$}\algorithmiccomment{Given $t$, fetch the corresponding partial sums}
				\State  $ i \gets \min\{j:s_j(t';r) \neq 0 \} $
				\State  $\hat{y}_t \gets \hat{y}_t  + \sum_{\ell = t'/r^i - s_j(t;r) + 1}^{t'/r^i } \left(p_{i,\ell} + Z_{i, \ell}\right)$
				\State  $t' \gets t' - s_i(t;r)\cdot r^{i} $
			\EndWhile
		\EndFor
		\State  Return $\hat{y}_{1:T} = \hat{y}_1, \cdots, \hat{y}_T$.	
	\end{algorithmic}\label{algo:BaseMechanism}
\end{algorithm}

We now state the privacy and utility guarantees of $\BinMech$. 
\begin{theorem}
For any base $r \in \{2, \cdots, T \}$, the $\BinMech(\sigma_{1:T}; \tau, r)$ is $\tfrac{1}{2\tau^2}$-zCDP.
\end{theorem}
\begin{proof}
The proof follows the same argument as in \cite{ChanShSo11}.  Rather than outputting the noisy counts, we instead consider outputting the entire table of partial counts $\bbp := \{ p(j,\ell) +Z_{j,\ell} : j \in [ L_r], \ell \in [T] \}$, where $\{Z_{j,\ell}\} \stackrel{i.i.d.}{\sim} \Normal{0}{L_r \tau^2}$.   Let $\sigma_{1:T}$ and $\sigma'_{1:T}$ be two neighboring streams with partial sum tables $\bbp = \{p_{j,\ell} \}$ and $\bbp' = \{ p_{j,\ell}' \}$, respectively.  Due to the way we defined neighbors, these two partial sum counts can differ in at most $L_r$ cells and can differ in each cell by at most 1.  Hence, the $\ell_2$-sensitivity of the partial sum table is at most $L_r$, and by adding $\Normal{0}{L_r \tau^2}$ to each cell's true count ensures $\frac{1}{2\tau^2}$-zCDP (see Lemma 2.5 in \cite{BunSt16}).  
\end{proof}

We now present the utility guarantee of the $\BinMech$ for any base $r$, which follows from tail bounds of Gaussian random variables.
\begin{theorem}
For any $r \in  \{2, \cdots, T \}$ and any time $t \in [T]$, the true count $y_t$ and $\hat{y}_t$ from $\BinMech(\sigma_{1:T};\tau,r)$ satisfies the following for any $\eta>0$
\begin{align*}
\Pr[|\hat{y}_t - y_t| \geq \eta]  \leq \Pr[|\Normal{0}{ (r -1)  L_r^2 \tau^2}| \geq \eta] 
\leq 2 \exp\left( \frac{-\eta^2}{2(r -1)  L_r^2 \tau^2} \right).
\end{align*}
\label{thm:BinMechUtility}
\end{theorem}
\begin{proof}
The first inequality in the lemma holds since in the worst case, $\BinMech(\sigma_{1:T}; \tau, r)$ will add at most $(r-1)L_r$ i.i.d samples from $\Normal{0}{L_r \tau^2}$. The second inequality holds since for any $\eta>0$, we have $\Pr[\Normal{0}{\sigma^2} > \eta] \leq \exp\left(-\frac{\eta^2}{2\sigma^2}\right)$. 
\end{proof}

Note that the base $r$, given $T$, can be selected in a way to minimize the overall variance of any single count, i.e.
\begin{equation}
r^* \defeq \argmin_{r \in \{2,\cdots, T \}}\left\{ (r - 1) L_r^2 \right\}.
\label{eq:opt_base}
\end{equation}
In Figure~\ref{fig:baseNoise} we plot the resulting standard deviation of noise with various bases $r$ and compare it with what we would get by using base $r = 2$ as in the original Binary Mechanism.  Note that it looks like we can reduce noise by about $15\%$ at the same level of privacy and for most practical settings $r^* \in \{3,\cdots,10 \}$.  Note that the optimal choice of $r$ is pretty stable, so that even if a gross upper bound $T$ is used on the event stream, the true optimal base will not change very much.  In our algorithms, we will keep the choice of base $r$ as arbitrary and remove its dependence in the later algorithms since it will not impact the privacy claims.
\begin{figure}
\centering
        \includegraphics[width=.5\textwidth]{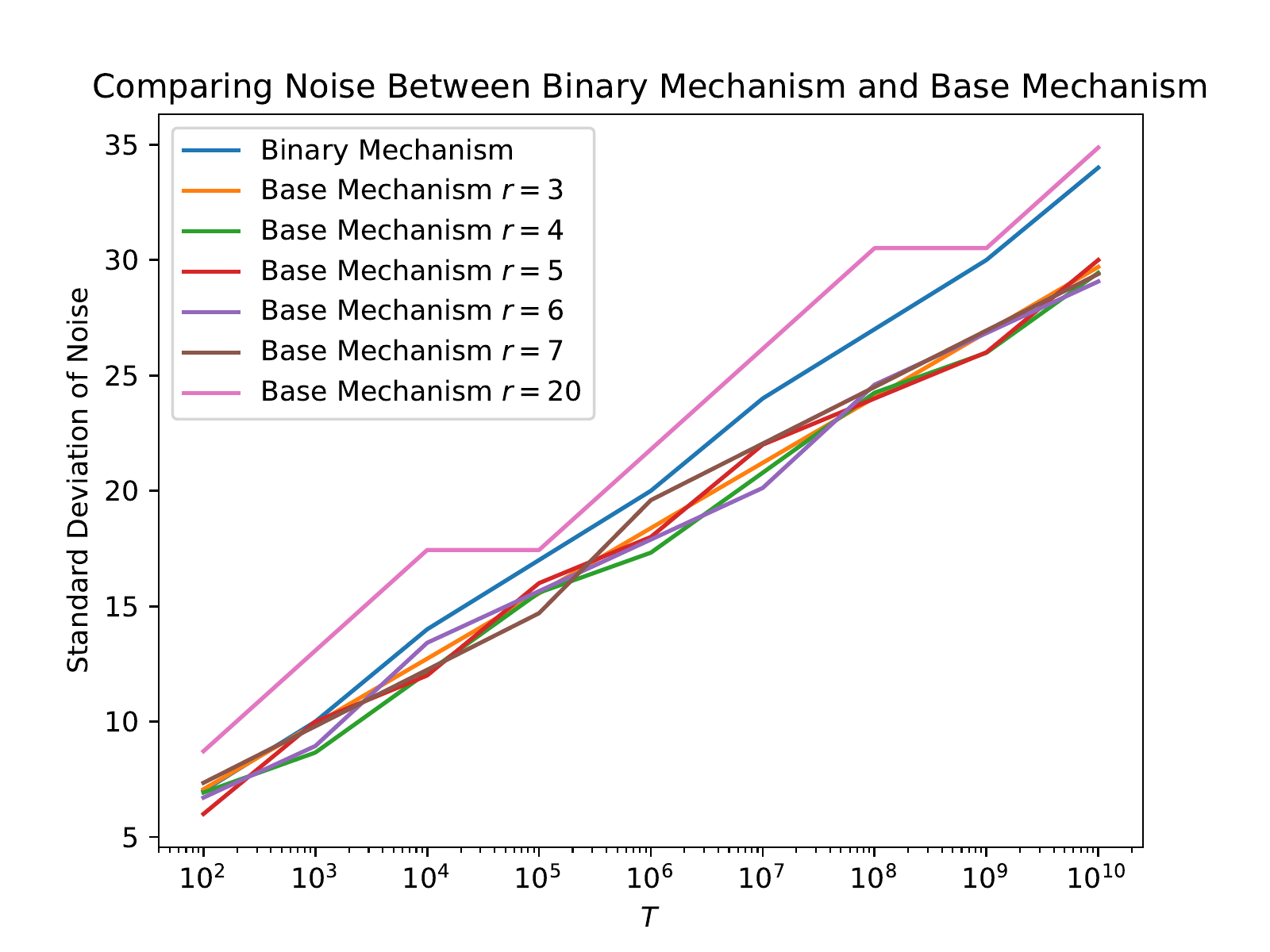}
    \caption{Comparison of the scale of the noise required in the classical Binary Mechanism and $\BinMech{}$ with various choices of base $r$.}
    \label{fig:baseNoise}
\end{figure}

In the case when it is known that a user can only modify a limited number of counts at each round $t$, i.e. $|\omega_t| \leq \Delta_0$ and the domain $\cU$ is known in advance, we can simply apply a stream of counts for each domain item.  This setting was considered in \citet{ChanShSo11}, and we provide the mechanism $\knownBase{}$ in Algorithm~\ref{algo:KnownBase}.
\begin{algorithm}[h!]
	\caption{$\knownBase$; Return a running histogram}
	\begin{algorithmic}
		\State  \textbf{Input:} $\omega_{1:T}$, with $\omega_t \subseteq \cU$, $r \in \{2, \cdots, T \}$, and $\tau$.
		\State  \textbf{Output:} Noisy histograms $\hat{\bbh}_{1:T} = \hat{\bbh}_1, \cdots, \hat{\bbh}_T$.
		\For{$u \in \cU$}
			\State  Define $\sigma_t^u = \1{u \in \omega_t }$ for all $t \in [T]$ \algorithmiccomment{Generate a stream of bits}
			\State  $\hat{h}_{1:T}^u = \BinMech(\sigma_{1:T}^u; \tau, r)$ \algorithmiccomment{Feed the stream of bits to $\BinMech$ }
		\EndFor
		\State  Return $\hat{\bbh}_{1:T} = \left( \left\{ (\hat{h}^u_{t}, u) : u \in \cU \right\} : t \in [T]\right)$.	
	\end{algorithmic}\label{algo:KnownBase}
\end{algorithm}

We then have the following privacy guarantee, which follows from the analysis in \citet{ChanShSo11} and composition of zCDP mechanisms \cite{BunSt16}.

\begin{lemma}
For streams $\omega_{1:T}$ such that $|\omega_t| \leq \Delta_0$ for each $t \in [T]$, $\knownBase(\cdot; \tau)$ is $\frac{\Delta_0}{2\tau^2}$-zCDP.
\end{lemma}
\begin{proof}
Let $\omega_{1:T}$ and $\omega_{1:T}'$ be two neighboring streams where there is a round $t$ where $\omega_t \neq \omega_t'$ where w.l.o.g. $|\omega_t| \leq \Delta_0$ and $\omega_t' = \emptyset$, while $\omega_{t'} = \omega_{t'}'$ for all $t' \neq t$.  Hence, there can be at most $\Delta_0$ many items $u \in \cU$ such that $\sigma^u_{1:T} \neq \sigma'^u_{1:T}$, while all other streams are identical.  Hence, we need only consider the total privacy of $\Delta_0$ many instances of $\BinMech$, which is each $\tfrac{1}{2\tau^2}$-zCDP.  Applying composition of zCDP mechanisms gives the result.
\end{proof}

%% file: metaAlgo.tex
\section{Meta Algorithm for Continual Observation\label{sect:meta}} 

In this section we propose a general scheme, $\metaAlgo$ in Algorithm~\ref{algo:MetaAlgo}, to return privatized histogram results in the various settings given in Table~\ref{table:tasks} but in the continual observation setting. Before describing the general scheme we briefly describe how the one-shot algorithms work. The known domain algorithms can be summarized as adding either Gaussian noise with standard deviation $\tau$, then returning the list of items and their counts or adding Gumbel noise with scale $\tau/2$ and taking the top-$k$ results then adding fresh Gaussian noise with standard deviation $\tau$ to those discovered items' counts.  Note that the Exponential Mechanism can be implemented by adding Gumbel noise to counts and then returning the element with the largest noisy count.  Further, the Exponential Mechanism with privacy parameter $\epsilon$ satisfies a property called \emph{bounded range} \cite{DurfeeRo19}, which results in $\tfrac{1}{8}\epsilon^2$-zCDP \cite{CesarRo20}. 

 The unknown domain algorithms can be thought of as the same as the known domain algorithms, except we only have access to the top-$(\bk+1)$ items from the full histogram and we include a noisy threshold that will depend on the privacy parameter $\delta$, so that only items above the noisy threshold will be shown. For completeness we present the pseudocode and privacy guarantees for each of the various algorithms in Appendix~\ref{appendix:oldAlgos}, except for $\unkGauss{\bk}$, which we analyze in a latter section.

The key observation is that we can generalize the partial sum table from Section~\ref{sect:knownBase} to a \emph{partial histogram table}  were each entry contains the histogram formed by the corresponding substream from $\omega_{1:T}$. Depending on what setting from the Table~\ref{table:tasks} we are in, we apply the corresponding one-shot DP algorithm to each cell $j,\ell$ for $j\in[\log_{r}(T)],\ell\in[T/{r}^{j-1}]$ of the partial histogram table, and aggregate the corresponding noisy histograms to provide a private result at time $t$.

\begin{algorithm}[h!]
	\caption{$\metaAlgo$; Return a running histogram in various settings from Table~\ref{table:tasks}}
	\begin{algorithmic}
		\State  \textbf{Input:} Stream $\omega_{1:T} = \omega_1, \cdots, \omega_T$, where $\omega_t\subseteq \cU$, and use $L_r$ from \eqref{eq:baseL}.
		\State  \textbf{Output:} Noisy histograms $\hat{h}_{1:T} = \hat{h}_1, \cdots, \hat{h}_T$ with labels for each count at each round.
		\If{Domain $\cU$ is known}
			\If{ $\ell_0$-Sensitivity $\Delta_0$}
				\For{ Each cell $j,\ell$ of the partial histogram table}\
					\State   $\{ (\hat{p}_{j,\ell}^u, u) \} = \knownGauss{}(\{ (p_{j,\ell}^u, u) : u \in \cU \}; \sqrt{L_r} \tau)$
				\EndFor
			\Else
				\For{ Each cell $j,\ell$ of the partial histogram table}
					\State   $\{ (\hat{p}_{j,\ell}^u, u) \} = \knownGumb{k}(\{ (p_{j,\ell}^u, u) : u \in \cU \}; \sqrt{L_r} \tau)$
				\EndFor
			\EndIf
		\Else
			\If{ $\ell_0$-Sensitivity $\Delta_0$}
				\For{ Each cell $j,\ell$ of the partial histogram table}
					\State  Let $\cU_{j,\ell}$ be the set of items  in the portion of $\omega$ used in $p_{j,\ell}$ 
					\State   $\{ (\hat{p}_{j,\ell}^u, u) \} = \unkGauss{\bk}(\{ (p_{j,\ell}^u, u) : u \in \cU_{j,\ell} \}; \sqrt{L_r} \tau,\delta/L_r)$
				\EndFor
			\Else
				\For{ Each cell $j,\ell$ of the partial histogram table}
					\State  Let $\cU_{j,\ell}$ be the set of items  in the portion of $\omega$ used in $p_{j,\ell}$ 
					\State   $\{ (\hat{p}_{j,\ell}^u, u) \} = \unkGumb{k,\bk}(\{ (p_{j,\ell}^u, u) : u \in \cU_{j,\ell} \}; \sqrt{L_r} \tau,\delta/L_r)$
				\EndFor
			\EndIf
		\EndIf
		\State  Return counts and labels at each round $t \in [T]$ using the  corresponding noisy partial histograms $\{ (\hat{p}^u_{j,\ell}, u) \}$.	
	\end{algorithmic}\label{algo:MetaAlgo}
\end{algorithm}

We state the various privacy guarantees in terms of the noise level $\tau$ and other parameters. The analysis follows by zCDP composition over at most $L_r$ cells that can change in the partial histogram tables of neighboring streams. 
\begin{theorem}
If  we have $\ell_0$-sensitivity $\Delta_0$ and known domain, then $\metaAlgo$ is $\frac{\Delta_0}{2\tau^2}$-zCDP.
If we have unrestricted $\ell_0$-sensitivity and known domain, then $\metaAlgo$ is $\frac{k}{\tau^2}$-zCDP.
If we have $\ell_0$-sensitivity $\Delta_0$ and unknown domain, then $\metaAlgo$ is $\left( \epsilon( \tfrac{\Delta_0}{2\tau^2}, \delta') ,\Delta_0 \delta + \delta' \right)$-DP for any $\delta'>0$, where $\epsilon(\cdot,\cdot)$ is given in \eqref{eq:epsilon}.
If we have unrestricted $\ell_0$-sensitivity and unknown domain, then $\metaAlgo$ is $\left( \epsilon(\tfrac{k}{\tau^2}, \delta') ,2k \delta + \delta' \right)$-DP for any $\delta'>0$, where $\epsilon(\cdot,\cdot)$ is given in \eqref{eq:epsilon}.
\end{theorem}
\begin{proof}
For completeness, we present all existing algorithms with their privacy guarantees in the appendix, except for $\unkGauss{\bk}$, which we cover in a later section.  The restricted $\ell_0$-sensitivity with known domain result follows from the zCDP analysis of the Gaussian Mechanism with $\ell_2$-sensitivity $\sqrt{\Delta_0\cdot L_r}$, due to an event changing at most $\Delta_0$ many counts by at most $1$ in at most $L_r$ many cells of the partial histogram table $\{ (p_{j,\ell}^u, u ) : u \in \cU \}$.  The unrestricted $\ell_0$-sensitivity with known domain result follows from applying the Exponential Mechanism to select $k$ items and then add Gaussian noise to each of the $k$ counts in each cell.  We then apply composition over $L_r$ cells of the table that can change when an event is changed.  

We have covered the analysis of $\unkGauss{\bk}$ in Theorem~\ref{thm:main_UnkGauss}, where we showed that we can separate good and bad outcome sets given a pair of neighboring datasets.  Note that bad outcomes are ones that can only occur in one neighboring dataset, which are only possible in the cells that can differ in neighboring streams.  Hence, we union bound all bad outcome sets over the $L_r$ cells, each of which has a probability of at most $\delta/L_r$.  For good outcomes in each cell, we can consider a specific Gaussian Mechanism from Algorithm~\ref{algo:GaussMech}.  Over the good outcomes, we are left with a Gaussian mechanism in each cell which can then be considered as a larger Gaussian mechanism with $\ell_2$-sensitivity $\sqrt{\Delta_0 L_r}$.  Lastly, we have the unrestricted $\ell_0$-sensitivity with unknown domain, whose analysis follows a similar argument to the \emph{pay-as-you-go composition}, although we always bound the number of exponential mechanisms to be at most $L_r\cdot k$.  Note that we then apply the Gaussian mechanism over the discovered items in each cell.
\end{proof}

The main drawback with this meta-algorithm is that in order to implement it at a time $t$, we will need to know the full stream of events, so that we can apply each DP algorithm on different subsequences, which then need to be stored for later calculations.  In latter sections we explore settings where our algorithms only have access to the aggregated histogram up to time $t$ at each round, rather than the full stream of events.

%% file: unrestrictedSensitivity.tex
\section{Unrestricted $\ell_0$-sensitivity, Known Domain\label{sect:sparseGumb}} 

We now consider the case where there is no limit to how many items a user can contribute for a given event in a stream, unlike in Section \ref{sect:knownBase} where the bound was $\Delta_0$.  To ensure there is some bound on privacy, we only display the top-$k$ results at round $t$, which are computed based on all events that have occurred in the stream up to that round.  This is particularly useful when no preprocessing of the data is in place to restrict the number of items for each event, yet we still want to ensure some bounded level of privacy, even for event level.  Otherwise, we would need to add noise that scales with $|\cU| = d$ due to users possibly contributing an arbitrary number of items.  

\subsection{Privacy Analysis}
Our algorithm consists of multiple classical DP algorithms, which makes the privacy analysis somewhat standard.  Consider the case when we want to return the top-1 item at every round. Given a data generating distribution, one would expect that the top-1 item would not change very many times in a stream of events. Hence, we introduce a parameter $s$, which is the number of \emph{switches} the algorithm is allowed to have.  A switch takes place when a new item has count significantly larger than the currently selected one.

\begin{algorithm}[h!]
	\caption{$\sparseGumb{s,k}$; Continually return top-$k$}
	\begin{algorithmic}
		\State  \textbf{Input:} $\omega_{1:T}$ with $s$, top-$k$ returned, $\{\eta_t \}_{t=1}^T$, and $\tau$.
		\State  \textbf{Output:} Noisy histograms $\hat{\bbh}_{1:T}$ for top-$k$ items.
		\State  Let $h_t^u$ be the count for item $u$ from $\bbh(\omega_{1:t};\cU)$.
		\State  Let $\tau_1 = \sqrt{s} \tau$ and $\tau_2 = \sqrt{s+1 }\tau$.
		\State  $\{i_1, \cdots, i_k \} = \knownGumb{k}( \bbh(\omega_{1:1};\cU); \tau_2)$.  \algorithmiccomment{Select top-$k$}		
		\State  Let $\sigma^{u}_{1:T}$ be the binary stream for item $u\in \{i_1, \cdots, i_k \}$.
		\For{$i \in \{i_1, \cdots, i_k \}$}
			\State  Get the current counts: $\hat{h}_{1:T}^i = \BinMech{}(\sigma^{i}_{1:T}; \tau_2)$ \algorithmiccomment{Return running counts of top-$k$}
		\EndFor
		\State  Set $\cD_1 = \{ i_1, \cdots, i_k\}$ and sample $Z \sim \lap(2 \tau_1)$.
		\For{$t \in \{ 2. \cdots,  T\}$}
			\If{$s = 0$}
				\State  {\bf break}  
			\EndIf
				\State  Set $i^* = \argmin_{i \in \{i_1, \cdots, i_k \}}\{ \hat{h}_t^{i}  \}$
				\State  Set threshold $\hat{m}_t = \hat{h}_t^{i^*} + \eta_t + Z$ \algorithmiccomment{Set threshold as in Sparse Vector}
				\For{$u \in \cU \setminus \{i_1, \cdots, i_k \}$}
					\If{$h_t^u + \lap(4\tau_1) > \hat{m}_t$} \algorithmiccomment{Check if there is a new element in the top-$k$.}
						\State  $\{i_1, \cdots, i_k \} = \knownGumb{k}(\bbh(\omega_{1:t}); \tau_2)$. \algorithmiccomment{Select top-$k$} 
						\For{$i \in \{i_1, \cdots, i_k \}$}
							\State  Set $\hat{h}_{1:T}^i = \BinMech{}(\sigma^{i}_{1:T}; \tau_2)$ \algorithmiccomment{Return running counts of the new top-$k$}
						\EndFor
						\State  $s \gets s-1$. \algorithmiccomment{Reduce by one the number of switches}
						\State  Redraw $Z \sim \lap(2\tau_1)$
						\State  {\bf break} 
					\EndIf
				\EndFor
			\State  $\cD_t = \{i_1, \cdots, i_k \}$
		\EndFor
		\State  Return $ \left( \left\{  (\hat{h}_t^u, u) : u \in \cD_t \right\} : t \in [T]\right)$
	\end{algorithmic}\label{algo:sparseGumb}
\end{algorithm}

We now discuss the algorithm at a high level.  At the first round, we will want to find the item with the top count, which can be done with the Exponential Mechanism \cite{McSherryTa07}, i.e. $\knownGumb{1}$ with only the items returned, not their counts.  Recall that we are in the known domain setting, so we will have the same domain at each round, which consists of $d$ items.  Finding this top item will cost a single unit of privacy in our composition, despite one user being able to have a set of items the size of the full domain $d$.  Once we have the top selected item, we can use the $\BinMech$ algorithm to produce a running count for this selected item.  However, we need to check the counts of other items at each round to see if there is one with higher count.  For this, we will use the Sparse Vector technique \cite{DworkNaReRoVa09} to continually check whether there is an item with larger count than the currently selected item.  We will only switch the top item if there is another item with count $\eta_t \geq 0$ more that the currently selected item's count at round $t$.  We can then set $\eta_t$ in our utility analysis, so we keep it arbitrary here.  Once we find that there is an item with larger count than the current top item, we will then use the Exponential Mechanism again to find a new top item, and continually return counts for the new item using $\BinMech{}$.  By the end of the stream, we know that there can be at most $s+1$ many items with counts from $\BinMech{}$. It is then easy to generalize this idea to allow for top-$k$ results at each round, rather than top-$1$. We call this generalization $\sparseGumb{s,k}(\omega_{1:T};\{\eta_t\}, \tau)$ and it has the following privacy guarantee.

\begin{theorem}
For any $\{\eta_t\}_{t=1}^T$ with $\eta_t \geq 0$ for all  $t \in [T]$,  $\sparseGumb{s,k}(\cdot;\{\eta_t\},\tau)$ is $\tfrac{2k+4}{2\tau^2}$-zCDP.
\end{theorem}
\begin{proof}
We rely on the privacy analysis of multiple subroutines.  We know that each call, of the $(s+1)$ calls, to the routine $\knownGumb{k}(\cdot,\sqrt{s+1}\tau)$, without releasing counts, is $\tfrac{k}{2(s+1)\tau^2}$-zCDP.  Further, there can be at most $(s+1) \cdot k$ many different instances of $\BinMech$, each of which is $\tfrac{1}{2(s+1) \tau^2}$-zCDP.  Lastly, we use the Sparse Vector technique to determine which rounds we should find a new top-$k$ in.  Note that we use different thresholds at each round $t$, but we do not update the noise on the threshold unless we update the top-$k$.  We then use the general version of Sparse Vector in \cite{LySuLi17} to conclude that each time we select a round to run $\knownGumb{k}$, it is $\tfrac{2}{\sqrt{s}\tau}$-DP and hence $\tfrac{2}{s\tau^2}$-zCDP.  We then apply composition of zCDP mechanisms to get the result.
\end{proof}

\subsection{Utility Analysis}

We now consider the utility of $\sparseGumb{s,k}$.  We will consider the case where $k = 1$.  Let $t_0,<t_1<\cdots < t_s  \leq T$ be the $s+1$ rounds that we select an element $i_{t_\ell}$ with $\knownGumb{k}$ for $\ell \in \{0,1,\cdots, s\}$.  Let $i_t^* \defeq \argmax\{ h_t^u : u \in [d]\}$ be the true max element at round $t \in [T]$.  We then calculate the error in $\sparseGumb{s,1}$ to be the following where we use $t_{s+1} = T+1$,
\[
\texttt{Err}(T) = \max_{\ell \in \{ 0, 1, \cdots, s\} } \left\{  \max_{t \in [t_{\ell}, t_{\ell+1})} \left\{ | \hat{h}_{t}^{i_{t_\ell}} - h_t^{i^*_t}|\right\} \right\}.  
\]

Due to a recent result from \cite{JainRaSiSm21}, we know a lower bound on the general streaming max problem (referred to as \texttt{SumSelect} in their work) is $\texttt{Err}(T)  = \Omega\left( \min\left\{ \sqrt[3]{\tfrac{T\log^2(d)}{\diffp^2}}, \tfrac{\sqrt{d}}{\diffp}, T \right\} \right)$ with high probability, so we instead consider non worse case streams to avoid the dependence on $\mathrm{poly}(T,d)$.  In particular, we will make an assumption on the stream of data.\footnote{Note that we are only making an assumption for utility and privacy holds in all cases.}. We first define a set of elements that are $\alpha>0$ close to the max count at round $t \in [T]$
\[
S_t(\alpha) \defeq \{u \in [d] : h_{t}^{u} \geq h_t^{i_t^*} - \alpha \}.
\]
Our assumption on the stream will involve three parameters, $\alpha_1< \alpha_2 < \alpha_3$.  At a high level, we will decompose $[T]$ into possibly overlapping intervals $B_0, A_1, B_1, A_2, \cdots, A_s, B_s$.  We will want to assume that the rounds in $A_\ell$ have a clear maximum element, or at least a cluster that is $\alpha_1$ close to the optimal and separated from elements outside the cluster.  Note that selecting a maximum element for the right choice of $\alpha_1$ will result in selecting an element from the cluster, with high probability.  Further, we define intervals $B_\ell$ that may overlap with $A_\ell$ and $A_{\ell+1}$, that will ensure no switch will occur, i.e. select a new maximum element, with high probability as long as all elements in from the cluster of elements in $A_\ell$ are not $\alpha_2$ smaller than the maximum element for the appropriate choice of $\alpha_2 > \alpha_1$. Lastly, we will want a subinterval $A'_\ell \subseteq A_\ell$ which will ensure with high probability that an element in the cluster is selected, i.e. the cluster is at least $\alpha_3$ larger than any element outside of the cluster. See Figure~\ref{fig:Assumption} for a picture showing the assumption for various times.

\begin{figure}
\centering
\includegraphics[width=0.75\textwidth]{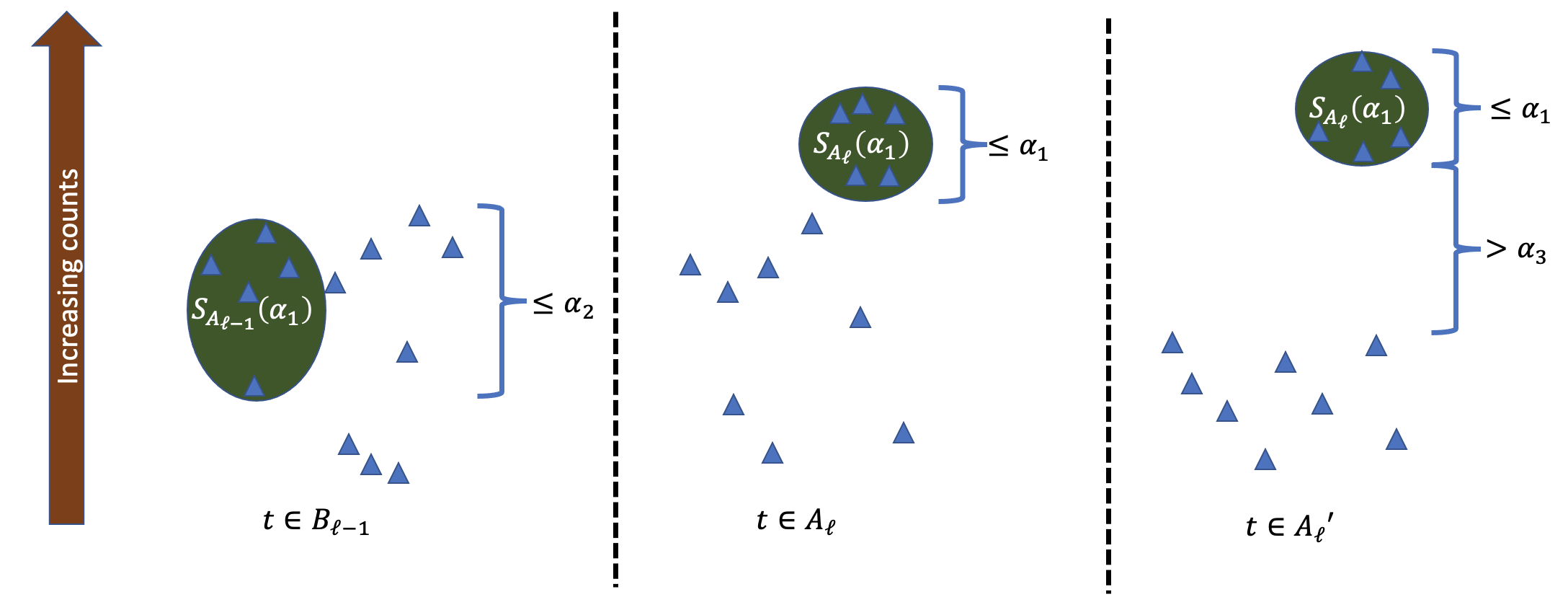}
\caption{A visualization of Assumption~\ref{assumpt:utility} for various rounds $t \in [T]$ with $\ell \in [s]$.  We denote with a triangle to be the counts of different elements in $[d]$.\label{fig:Assumption}}
\end{figure}

This assumption closely matches the expected behavior of sampling i.i.d. from a distribution with a clear maximum element for some time and then switching to sampling i.i.d. from another distribution with a different maximum element.  We now state the assumption more formally.
\begin{assumption}[Utility]\label{assumpt:utility}
Let $\alpha_1, \alpha_2, \alpha_3 >0$ be three parameters such that $\alpha_1 < \alpha_2 < \alpha_3$.  We will decompose the set $[T]$ into different intervals $B_0, A_1, B_1, \cdots, A_s, B_s$ that cover all of $[T]$ which will have the following conditions for each $\ell \in [s]$:
\begin{itemize}
\item For all $t \in B_0$, we have $h_t^u \geq h_t^{v} - \alpha_2$ for each $u, v \in [d]$.
\item There exists a set $S_{A_\ell}(\alpha_1) \subseteq [d]$ such that $S_t(\alpha_1) = S_{A_\ell}(\alpha_1)$ for each $t \in A_\ell$.
\item For each $t \in A_\ell$ and each $u \in S_{A_\ell}(\alpha_1)$, we have 
\[
h_t^u > h_t^v, \quad \forall v \notin S_{A_\ell}(\alpha_1).
\]
\item For all $t \in B_\ell$, and for all $u \in S_{A_\ell}(\alpha_1)$ and for $v \in [d]\setminus \{u \}$, we have
\[
h_t^u \geq h_t^{v} - \alpha_2.
\]
\item There exists a nonempty subinterval $A'_\ell \subseteq  A_\ell$ such that for all $t \in A'_\ell$ and each $u \in S_{A_\ell}(\alpha_1)$, we have
\[
h_t^u \geq h_t^v + \alpha_3, \quad \forall v \notin S_{A_\ell}(\alpha_1).
\]
\end{itemize} 
\end{assumption}
With Assumption~\ref{assumpt:utility}, we will work up to the full utility statement by adding each privacy mechanism one by one, starting with the Exponential Mechanism.
\begin{lemma}\label{lem:expMechUtility}
Given Assumption~\ref{assumpt:utility}, for any $\ell \in [s]$, if we run the exponential mechanism in $\sparseGumb{s,1}$ within any $A_\ell$ then with probability at least $1-\beta$, the index at each of those rounds will be in $S_{A_\ell}(\alpha_1)$ where
\begin{equation}
\alpha_1 = \Theta\left( \sqrt{s}\tau \log(d/\beta)\right).
\label{eq:alpha1}
\end{equation}
\end{lemma}
\begin{proof}
This follows from the utility theorem of the Exponential Mechanism \cite{McSherryTa07}, which we write out here.  Let $i_t$ be the index selected from the Exponential Mechanism $\knownGumb{1}( \{ h^u_t, u : u \in \cU\} ; \tau)$, without its count, and let $i^*_t$ be the argmax of $\{ h^u_t: u \in \cU\}$.  Then for $\beta > 0$, we have
\[
\Pr[h^{i_t}_t > h^{i^*_t}_t - \tau \ln\left(d/\beta \right) ] \geq 1 - \beta.
\]
\end{proof}

We now consider the sparse vector component, also referred to as AboveThreshold in \cite{DworkRo14}, which will never return a new top element in $B_\ell$ for $\ell \in [s]$ if $\eta = \eta_t$ and $\alpha_1,\alpha_2, \alpha_3$ are set appropriately.
\begin{lemma}\label{lem:SVTutility}
Given Assumption~\ref{assumpt:utility}, we will use $\alpha_1$ from \eqref{eq:alpha1} and 
\[
\alpha_2 = \Theta\left( \alpha_1 \right), \alpha_3 = O\left(  \tau \sqrt{s} \log^{3/2}(dT/\beta)\right), \eta = O\left(  \tau \sqrt{s} \log^{3/2}(dT/\beta)\right).
\]
Let $i_0 \in [d]$ be the index selected initially in $\sparseGumb{s,1}$.  With probability at least $1-\beta$ over all rounds $t \in B_0$ we will not select a new maximum element and there will be a round $t_1 \in A_1$ where we either select a new element $i_1 \in S_{A_1}(\alpha_1)$ or $i_0$ is already in $S_{A_1} (\alpha_1)$, and we relabel $i_1 = i_0$.  

Further, for $\ell \in [s]$, conditioning on $i_{\ell-1} \in S_{A_{\ell-1}}(\alpha_1)$ at some round $t_{\ell-1} \in A_{\ell-1}$ we have that for all $t> t_{\ell-1}$ and $t \in A_{\ell-1} \cup B_{\ell-1} \cup A_\ell$, we will not select a new element until some round $t_\ell \in A_\ell$ where $t_\ell$ is at most the first time in $A_\ell'$ with probability at least $1-\beta$. 

\end{lemma}
\begin{proof}
We first need to ensure that the noisy counts are within some error bound of the true counts.  From Theorem~\ref{thm:BinMechUtility}, we have that for any fixed $u \in [d]$ the following holds with probability $1-\beta/3$,
\begin{equation}
\max_{t \in (t_\ell, t_{\ell+1})} | \hat{h}_t^u - h_t^u| \leq \tau L_r\sqrt{2 (s+1)(r-1) \log(6T/\beta)} \eqdef \alpha_{\texttt{BM}}.
\label{eq:alpha''}
\end{equation}
We will condition on the case where each noisy count is within $\alpha_{\texttt{BM}}$ of its true count.  We also know from Lemma~\ref{lem:expMechUtility} that the selected index $i_{t_\ell}$ will be within $\alpha_1 = O\left( \sqrt{s}\tau \log(d/\beta)\right)$ of the true max with probability $1-\beta/3$, which we will assume in the remainder of the proof.

From the utility guarantee of the AboveThreshold algorithm \cite{DworkRo14}, we have that as long as $h_t^u \leq M_t - \alpha_{\texttt{AT}}$ for all $u \in [d]\setminus \{ i_{t_\ell}\}$  and $t \in B_{\ell-1}$ for some threshold $M_t$ yet to be determined, then AboveThreshold will return $\bot$ for all rounds in interval $B_{\ell-1}$.  Further, we have that for any $t \in A'_{\ell}$ and assuming $i_{t_{\ell-1}} \notin S_{A_{\ell}}(\alpha_1)$, then $h_{t}^u > M_{t} + \alpha_{\texttt{AT}}$ for some $u \in [d] \setminus \{ i_{t_{\ell-1}}\}$ then AboveThreshold will return $\top$ for some round in $A_{\ell}$ with probability at least $1-\beta/3$ where
\begin{equation}
\alpha_{\texttt{AT}} \defeq 8\tau \sqrt{s} \log(6 d T/\beta) .
\label{eq:alpha'''}
\end{equation} 
If $i_{t_{\ell-1}} \in S_{A_{\ell}}(\alpha_1)$, then we would not need AboveThreshold to return $\top$, since it has already selected an element near the next maximum.  Note that in $\sparseGumb{s,1}$, we have $M_t = \hat{h}_t^{i_{t_\ell}} + 
\eta$ where $t \in [t_\ell, t_{\ell+1})$.

We now need to determine $\eta$ based on Assumption~\ref{assumpt:utility}.  Let's start with $t \in B_0$.  The element $i_{t_0}$ that we select at some round $t_0 \in B_0$ will be within $\alpha_2$ of all other elements throughout $B_0$.  Hence we have for all $u \neq i_{t_0}$ and $t \in B_0$
\[
h_t^u \leq h_t^{i_{t_0}} + \alpha_2 \leq \hat{h}_t^{i_{t_0}} + \alpha_{\texttt{BM}} + \alpha_2. 
\]
We then set $\eta$ in $\sparseGumb{}$ as the following,
\begin{equation}
\eta \defeq \alpha_2 + \alpha_{\texttt{BM}} + \alpha_{\texttt{AT}}
\label{eq:eta}
\end{equation}
Note that if there exists a time $t_\ell \in A_\ell$ where AboveThreshold returns $\top$, then the exponential mechanism will select $i_{t_\ell} \in S_{A_\ell}(\alpha_1)$, since by hypothesis $h_t^u > h_t^v$ for all $u \in S_{A_\ell}(\alpha_1)$ and $v \notin S_{A_\ell}(\alpha_1)$.  We want to show that there will actually be a round $t \in A_\ell$ where we will have the selected element $i_{t_\ell} \in S_{A_\ell}(\alpha_1)$.  In particular, we know there is a round $t \in A'_\ell$ where for all $u \in S_{A_\ell}(\alpha_1)$ and $v \notin S_{A_\ell}(\alpha_1)$,
\[
h_t^u \geq h_t^v + \alpha_3
\]
If the previously selected element $i_{t_{\ell-1}} = v \notin S_{A_\ell}(\alpha_1)$, then 
\[
h_t^u \geq h_t^{i_{t_{\ell-1}}} + \alpha_3 \geq \hat{h}_t^{i_{t_{\ell-1}}} - \alpha_{\texttt{BM}} + \alpha_3 = \hat{h}_t^{i_{t_{\ell-1}}} + \eta - 2\alpha_{\texttt{BM}} - \alpha_2 - \alpha_{\texttt{AT}} + \alpha_3
\]
Hence, we need to ensure that following holds, so that with probability at least $1-\beta/3$ we will select something in $S_{A_\ell}(\alpha_1)$
\[
- 2\alpha_{\texttt{BM}} - \alpha_2 - \alpha_{\texttt{AT}} + \alpha_3 \geq \alpha_{\texttt{AT}}.
\]
\begin{equation}
\implies \alpha_3 - \alpha_2 \geq 2 \alpha_{\texttt{BM}} + 2 \alpha_{\texttt{AT}}
\label{eq:alphaGap}
\end{equation}
Note that once we have selected a new element in $i_{t_{\ell}} \in S_{A_\ell}(\alpha_1)$, we have that for each $t \in A_{\ell}$ where $t > t_\ell$ and each $u \neq i_{t_\ell}$
\[
h_t^u \leq h_t^{i_{t_\ell}} + \alpha_1 \leq \hat{h}_t^{i_{t_\ell}} +  \alpha_{\texttt{BM}} + \alpha_1 = \hat{h}_t^{i_{t_\ell}} +  \eta + \alpha_1 - \alpha_2 - \alpha_{\texttt{AT}}
\] 
Thus, because $\alpha_2 > \alpha_1$, we know that AboveThreshold will only return $\bot$ in the rest of $A_\ell$.

Given $i_\ell \in S_{A_{\ell}}(\alpha_1)$, we now turn to rounds $t \in B_\ell$.  By assumption, we have for each $v \neq i_{t_\ell}$
\[
h_t^v \leq h_t^{i_{t_\ell}} + \alpha_2 \leq \hat{h}_t^{i_{t_\ell}} + \alpha_{\texttt{BM}} +  \alpha_2 = \hat{h}_t^{i_{t_\ell}} + \eta - \alpha_{\texttt{AT}}.
\]

Hence, we use 
\[
\alpha_2 = \Theta(\alpha_1) = \Theta(\sqrt{s}\tau \log(d/\beta))
\]
 and 
 \[
 \alpha_3 = \Theta(\alpha_{\texttt{BM}} + \alpha_{\texttt{AT}}) =  \Theta\left( \tau \sqrt{s} \log^{3/2}(dT/\beta)\right).
 \]
 This makes $\eta = \Theta\left( \tau \sqrt{s} \log^{3/2}(dT/\beta)\right)$.
\end{proof}

With our technical lemmas, we are now ready to prove our main utility result.

\begin{theorem}
Given Assumption~\ref{assumpt:utility} and setting $\alpha_1, \alpha_2, \alpha_3, \eta$ as in Lemmas~\ref{lem:expMechUtility} and \ref{lem:SVTutility} we have that with probability at least $1 - s \beta$, 
\[
\texttt{Err}(T) = O\left(\tau \sqrt{s} \log^{3/2}(dT/\beta) \right)
\]
\end{theorem}
\begin{proof}
We apply a union bound in Lemma~\ref{lem:SVTutility} to ensure that at there are $s$ times $t_\ell \in S_{A_\ell}(\alpha_1)$ where we will select an element $i_{t_\ell} \in A_{\ell}$ and then not again in the remainder of $A_\ell$ and in $B_\ell$, where $\ell \in [s]$ .  We now consider the error in the counts.  At each round $t \in A_\ell$ in which we select an element $i_{t_\ell}$, call this round $t_\ell$, or after, we have with probability at least $1 - s \beta$ that 
\[
|\hat{h}_{t}^{i_{t_\ell}} - h_{t}^{i_{t}^*} | \leq |h_{t}^{i_{t_\ell}} - h_{t}^{i_{t}^*}| + \alpha_{\texttt{BM}}  \leq  \alpha_1 + \alpha_{\texttt{BM}} \leq  O\left( \tau \sqrt{s} \log^{3/2}(dT/\beta) \right).
\]
Further, for all times $t \in A_{\ell}\setminus A_\ell'$ before we select a new element at round $t_{\ell}$, we have
\[
\hat{h}_t^{i_{t_{\ell-1}}} - \alpha_{\texttt{BM}} \leq h_t^{i_{t_{\ell-1}}}\leq h_t^{i_{t}^*} \leq h_t^{i_{t_{\ell-1}}} + \alpha_3 \leq \hat{h}_t^{i_{t_{\ell-1}}}  + \alpha_{\texttt{BM}} + \alpha_3.
\]
Hence, we have for all $t < t_{\ell}$ and $t \in A_{\ell}\setminus A_\ell'$, we have
\[
|\hat{h}_t^{i_{t_{\ell-1}}} - h_{t}^{i_{t}^*}| \leq \alpha_{\texttt{BM}} + \alpha_3  = O\left( \tau \sqrt{s} \log^{3/2}(dT/\beta) \right) 
\]
We also consider all $t \in B_{\ell}$, in which case we have a similar condition as above,
\[
\hat{h}_t^{i_{t_\ell}} - \alpha_{\texttt{BM}} \leq h_t^{i_{t_\ell}} \leq h_t^{i_{t}^*} \leq \hat{h}_t^{i_{t_\ell}} + \alpha_{\texttt{BM}} + \alpha_2
\] 
This completes the proof.
\end{proof}

We will conduct experiments to see how the number of switches $s$ and $\{\eta_t \}_{t=1}^T$ impact the accuracy of the current round's selected item and the true maximum count. 
We will also need to set the additional threshold amounts $\eta_t$ for each round $t \in [T]$.  We will try several values of $\eta_t = \eta$ in our experiments.

\begin{figure}
\centering
\includegraphics[width=0.45\textwidth]{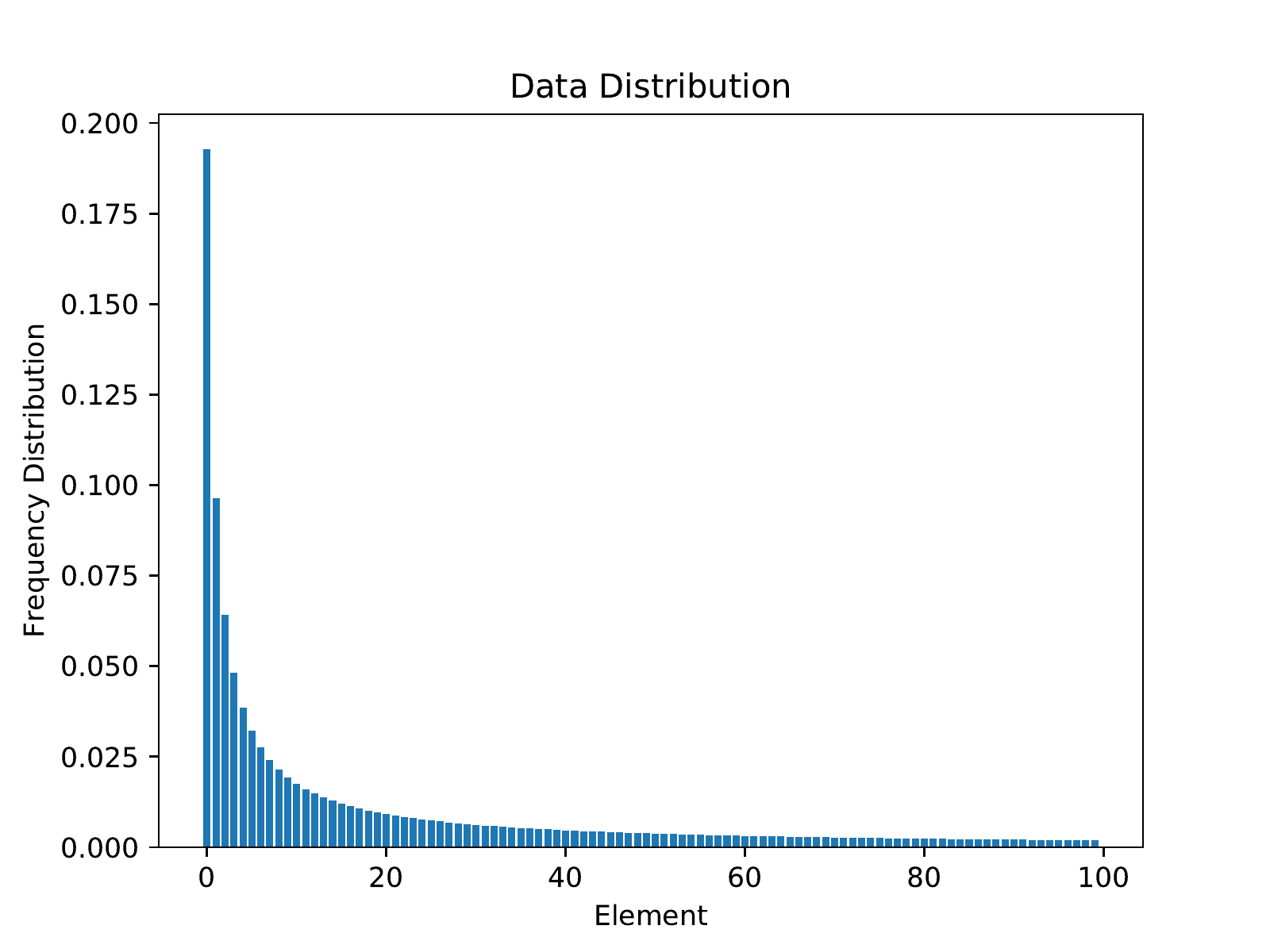}
\includegraphics[width=0.45\textwidth]{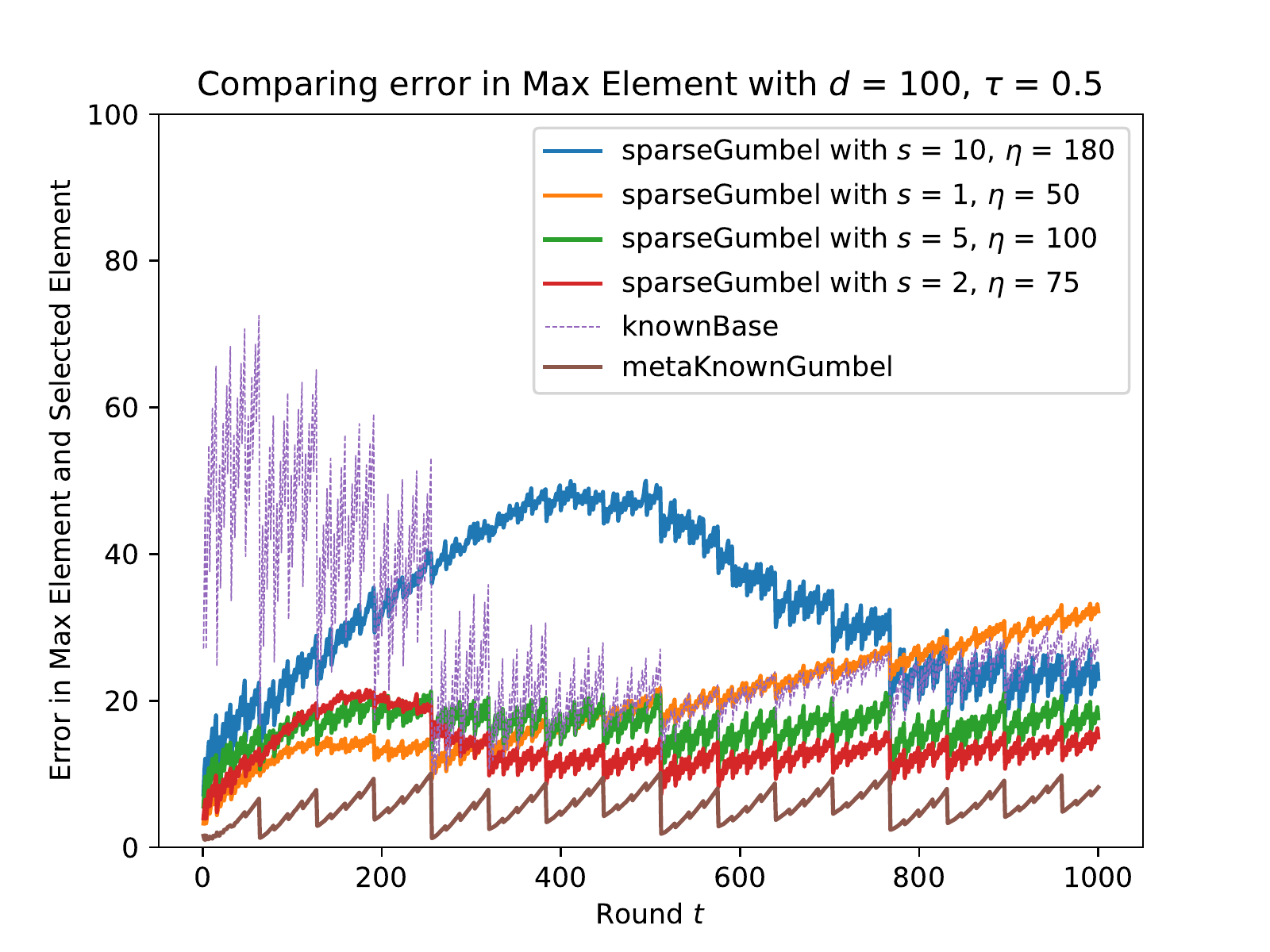}
\caption{(Left) We give the Zipf's Law distribution that items are sampled from at each round of an event stream. (Right)
We compare $\sparseGumb{s,1}$ with $\knownBase{}$ and $\metaAlgo$ with unrestricted $\ell_0$-sensitivity, denoted as `metaKnownGumbel' in the plot.\label{fig:sparseGumb}}
\end{figure}

To generate a stream of data, we sample an item from a distribution following Zipf's Law, as it models many data sources that occur in nature well, given in the left plot of Figure~\ref{fig:sparseGumb}, with $d=100$ items.  The right plot in Figure~\ref{fig:sparseGumb} shows the error between the count of the true max item with the noisy count of the selected top-$1$ item $\sparseGumb{s,1}$ at each round $t \in [1000]$ with various $s$ and $\eta \equiv \eta_t$.   We compare this algorithm with both the $\knownBase{}$ algorithm, where we use $\Delta_0 = d$, since we are assuming unrestricted $\ell_0$-sensitivity, and we also compare the results with $\metaAlgo$ in the same setting.  In our experiments, we will equalize the privacy level in all algorithms.  Hence, we will use $\tau \gets \tau \sqrt{d}$ in $\knownBase{}$, $\tau \gets \tau\sqrt{2}$ in $\metaAlgo{}$, and $\tau \gets \tau \sqrt{6}$ in $\sparseGumb{s,1}$, so that each will be $\frac{1}{2\tau^2}$-zCDP. 
As expected $\metaAlgo$ outperforms the other algorithms, but recall that at each round $t$ it needs the full stream $\omega_{1:t}$, which may be impractical in some situations (see Section~\ref{sect:intro}). Instead, $\sparseGumb{}$ only requires the current aggregate histogram. It is interesting to notice the behavior of $\sparseGumb{}$ with respect to $s$, with very few switches ($s=1$) the algorithm runs out of switches before the maximum element is learned and thus the error seems to increase linearly. If we allow $\sparseGumb{}$ more switches, it does not run out of switches very quickly, unfortunately the magnitude of the noise scales with $\sqrt{s}$ thus hurting accuracy. The right number of switches is a parameter that the practitioner needs to tune to balance the amount of noise incurred and the number of times the distribution is expected to change. The plots show the average error at each $t \in [1000]$ over 1000 independent trials.

%% file: revisitUnkGauss.tex
\section{Revisiting One-Shot Restricted $\ell_0$-sensitivity, Unknown Domain \label{sect:revisit}} 

In this section we present a new one-shot algorithm, $\unkGauss{\bk}$ for the restricted $\ell_0$-sensitivity and unknown domain setting. We first point out that in the one-shot case, we are considering only having access to a limited number of elements from the true histogram, in particular the top-$(\bk+1)$ elements, rather than the full histogram.  This is particularly useful for case when DP algorithms can only be applied to an aggregated result that is limited by how many elements can be fetched from the original dataset, see for example the setting in \cite{RogersSuPeDuLeKaSaAh20}.  When the full histogram is available yet the labels of the elements are unknown, algorithms from \cite{KorolovaKeMiNt09} and \cite{GoogleDPSQL} can be used.

The first algorithm for this limited histogram setting $\texttt{LimitDom}_\lap$ was developed by \citet{DurfeeRo19}. Through a new analysis (which we will also use in Section~\ref{sect:UnkBase}) we show that $\unkGauss{\bk}$ attains a better privacy guarantee than $\texttt{LimitDom}_\lap$ with the same level of noise. The algorithm is simple; given access to the $(\bk+1)$ highest ranked elements in the histogram, it adds Gaussian noise to each element and releases only those with noisy counts above a threshold $h_\bot \defeq h_{(k+1)} + 1 + \sqrt{2} \tau \Phi^{-1}(1-\delta)$ with $\delta>0$, which we label as $\bot$, and also has noise added to it. x

In both analyses, the neighboring datasets are given. From the neighboring datasets, the set of outcomes can be partitioned into good outcomes (those that can occur with both datasets) and bad outcomes (those that can only occur in one dataset).  However, notice that the set of elements that we add noise to in either dataset differs, since the top-$\bk$ in one dataset might be different than the top-$\bk$ in the other.  The earlier analysis consisted of applying the Laplace Mechanism only over the common elements in both datasets and showing that the probability of any good outcome from this Laplace mechanism over common elements is within $\delta$ of the probability of any good outcome from the Laplace mechanism over the full top-$\bk$ elements. Unfortunately, this resulted in a final $(\Delta_0\epsilon,\delta + e^{\epsilon \Delta_0} \delta)$-DP guarantee.  We will show that with the same amount of noise, we can achieve better privacy guarantees, without the $e^{\epsilon\Delta_0}$ factor on the $\delta$ term, by relabeling elements that cannot be released in both datasets.

\begin{algorithm}[h!]
	\caption{$\unkGauss{\bk}$; $\Delta_0$-Restricted Sensitivity Gaussian Mechanism with top-$(\bk+1)$}
	\begin{algorithmic} 
		\State \textbf{Input:} Histogram $\bbh = \{(h_u,u) : u \in \cU \}$, cut off at $\bk$, along with parameters $\tau,\delta$.
		\State \textbf{Output:} Noisy histogram with labels $\{i_j\}$ and noisy counts $\{v_{i_j}\}$.
		\State Let $h_{i_{(1)}} \geq h_{i_{(2)}} \geq \cdots  \geq h_{i_{(\bk)}} \geq h_{i_{(\bk+1)}} \geq \cdots \geq h_{i_{(d)}}$, with corresponding labels $i_{(j)} \in \cU$ for $j \in [d]$
		\State Set $v_\bot = h_{i_{(\bk + 1)}}+ 1 +  \sqrt{2}\tau \Phi^{-1}(1 - \delta)  + \Normal{0}{\tau^2}$ \algorithmiccomment{Set (data-dependent) noisy threshold}
		\State Set discovered set $\cD = \emptyset$ 
		\For{$u \in \cU$ such that $u \in \{i_{(j)} : j \in [\bk] \}$ and $h_u > 0$}\algorithmiccomment{Add noise to each element in top-$\bk$}
			\State Set $v_{u} =  \Normal{h_{u}}{\tau^2}$ with label $u$.
			\State $\cD \gets \cD \cup \{u \}$
		\EndFor
        \State Sort $\{v_u : u \in \cD\} \cup v_\bot$
        \State Let $v_{i_{(1)}},....,v_{i_{(\ell)}}$ be the counts in descending order until $v_\bot$, with relative labels $i_{(1)}, i_{(2)}, \cdots, i_{(\ell)} $.
		\State Return $\{(v_{i_{(1)}}, i_{(1)} ),...,(v_{i_{(\ell)}}, i_{(\ell)} )\}$ \algorithmiccomment{Return elements above the noisy threshold}		
	\end{algorithmic}\label{algo:UnkGauss}
\end{algorithm}

\begin{theorem}
\label{thm:main_UnkGauss}
For histograms with $\ell_0$-sensitivity $\Delta_0$ and $\ell_\infty$-sensitivity $1$, $\unkGauss{\bk}(\cdot; \tau,\delta)$ is $(\epsilon(\tfrac{\Delta_0}{2\tau^2},\delta'),\Delta_0\delta + \delta')$-DP for any $\delta' > 0$ with $\epsilon(\cdot,\cdot)$ in \eqref{eq:epsilon}.
\end{theorem}

We now describe the proof technique used to prove Theorem~\ref{thm:main_UnkGauss}, which will also be used in Section~\ref{sect:UnkBase} to analyze our more practical DP algorithm for the continual observation setting with $\ell_0$-sensitivity $\Delta_0$ and unknown domain.  We first set up some notation.  Let $M$ be a mechanism that takes input datasets $\bbh$ to some arbitrary outcome space.  For any two datasets $\bbh^{(0)}$ and $\bbh^{(1)}$, we define the \emph{good} outcome sets $\cG_M$, as outcomes that can occur with input $\bbh^{(0)}$ and $\bbh^{(1)}$ and the bad outcome sets $\cB_M^{b}$ for $b \in \{0,1\}$, as outcomes of $M$ that can occur with input $\bbh^{(b)}$ but not $\bbh^{(1-b)}$.  

The following result allows us to determine the privacy of a particular mechanism by analyzing the privacy of a related mechanism with access to both neighboring datasets.

\begin{lemma}\label{lem:meta}
Let $\bbh^{(0)}$ and $\bbh^{(1)}$ be two neighboring datasets.  Suppose there exists a mechanism $A(b;\bbh^{(0)}, \bbh^{(1)})$ where $b \in \{ 0,1\}$ such that for any outcome set $S \subseteq \cG_M$, we have
$
\Pr[M(\bbh^{(b)}) \in S] = \Pr[A(b;\bbh^{(0)}, \bbh^{(1)}) \in S].
$
Further, suppose that $\Pr[M(\bbh^{(b)}) \in \cB^b_M] \leq \delta$ for $b \in \{ 0,1\}$.
If $A(\cdot ; \bbh^{(0)}, \bbh^{(1)})$ is $(\epsilon,\delta')$-DP, then $M$ is $(\epsilon, \delta + \delta')$-DP.
\end{lemma} 
\begin{proof}
Fix an outcome set $S$, we then have
\begin{align*}
\Pr[M(\bbh^{(b)}) \in S] & = \Pr[M(\bbh^{(b)}) \in S\cap \cG_M] + \Pr[M(\bbh^{(b)}) \in S \cap \cB_M^b] \\
& \leq \Pr[M(\bbh^{(b)}) \in S\cap \cG_M] +\delta \\
& = \Pr[A(b;\bbh^{(0)}, \bbh^{(1)}) \in S \cap \cG_M] + \delta \\
& \leq e^{\epsilon} \Pr[A(1-b;\bbh^{(0)}, \bbh^{(1)}) \in S \cap \cG_M]  + \delta' + \delta \\
& = e^{\epsilon} \Pr[M(\bbh^{(1-b)}) \in S \cap \cG_M] + \delta' + \delta\\
& = e^{\epsilon} \Pr[M(\bbh^{(1-b)}) \in S] + \delta' + \delta
\end{align*}
\end{proof}

Hence, to prove the privacy of $\unkGauss{\bk}$, we show that bad outcomes occur with negligible probability and that there is a mechanism on shared outcomes of neighboring datasets that is DP.   Note that the parameter $\bk$ in $\unkGauss{\bk}$ means that we only have the top-$(\bk+1)$ elements available from the original histogram.  It might be the case that $\bk$ is larger than the number of elements in the histogram that actually have positive count.  Hence, $\unkGauss{\bk}$ might add noise to fewer than $\bk +1$ elements.  

Consider a slight variant of $\unkGauss{\bk}$, which we denote as $\unkGaussTop{\bk}$, that pads the histogram with zero counts and dummy labels $\{\top_i \}$ to ensure that there are exactly $\bk + 1$ many elements to add noise to. The next lemma shows that adding noise to \emph{dummy} elements but then dropping those elements from the outcome is the same as simply not even considering these dummy elements to begin with. 

\begin{lemma}\label{lem:dummy}
Let  $\bbh = \{ (h_u, u) : u \in \cU \}$ be a histogram with labels for $|\cU| = p$ elements and $M(\bbh)$ return $\bk+1$ counts with labels for $\bk\geq p$ with noise from some distribution $\cP$ where $\{ Z_i \} \stackrel{i.i.d.}{\sim} \cP$ and
\begin{align*}
M(\bbh) = \{ &(h_u + Z_u, u) : u \in \cU \} \cup \{ (h_{\bot} + Z_{\bot}, \bot) \} \cup \\
 &\{ (Z_{\top_j}, \top_j ) : j \in \{ 1, \cdots, \bk - p \} \}.
\end{align*}
Let $M'(\bbh)$ drop elements with counts lower than $\bot$ and then drop any element with label in $\{\top_i\}$.  Let $\hat{M}$ be the mechanism that adds i.i.d. noise from $\cP$ to only counts in $\bbh$ and $\hat{M}'$ drop elements with counts lower than the count labeled $\bot$. Then $M'(\bbh)$ is equal in distribution to $\hat{M}'(\bbh)$. 
\end{lemma}
\begin{proof}
We need to show that adding independent noise to $\bk$ counts, some of which have $\{\top_i\}$ labels and then dropping these terms is equivalent to having never considered those elements.

Let $f(\cdot)$ be the density function for distribution $\cP$, $f_{M'}(\cdot)$ be the density of $M'$, and $f_{\hat{M}'}$ be the density of $\hat{M}'$.  We fix an outcome of counts  $(z_1, z_2, \cdots, z_k)$ and denote the set of indices that are not in this outcome to be $I$ after dropping counts of $\{ \top_i \}$. The density for mechanism $\hat{M}'$ is then  
\begin{align*}
f_{\hat{M}'}(z_1, \cdots, z_k) & =  \prod_{i=1}^k f(z_i - h_i) \int_{-\infty}^{ \min\{z_i:i \in [k]\}} f\left(z_\bot - h_\bot\right) \prod_{\ell \in I }\left( \int_{-\infty}^{z_\bot} f(z_\ell - h_\ell)  dz_\ell \right) dz_{\bot}    \\
& =  \prod_{i=1}^k f(z_i - h_i) \int_{-\infty}^{ \min\{z_i:i \in [k]\}} f\left(z_\bot - h_\bot\right)\prod_{\ell \in I }\left( \int_{-\infty}^{z_\bot} f(z_\ell - h_\ell)  dz_\ell \right)dz_{\bot}   \\
& \qquad \qquad \cdot \prod_{j = 1}^{\bk - p - 1}\int_{\R} f(z_j) dz_{j} \\
& =  f_{M'}(z_1, \cdots, z_k).
\end{align*}
\end{proof}

Therefore, we prove the privacy of $\unkGaussTop{\bk}$, rather than $\unkGauss{\bk}$, since the latter is equal in distribution to a post-processing function of the former and cannot increase the privacy loss of $\unkGaussTop{\bk}$.  Our privacy analysis consists of analyzing the Gaussian Mechanism and bounding bad events, i.e. events that cannot occur in both neighboring histograms.  We define two domains of labels from a given histogram $\bbh = \left\{ (h_u,u) : u \in \cU \right\}$ with ordered indices $h_{i_{(1)}}\geq h_{i_{(2)}} \geq \cdots \geq h_{i_{(d)}}$.  The first only considers elements with positive count and the second pads the domain with zero counts and \emph{dummy} labels:
\begin{align*}
\cD^{\bk}(\bbh) & \defeq \{ i_{(j)} \in \cU : j \leq \bar{k} \text{ and } h_{i_{(\bk)}} > 0 \} \cup \{ \bot\} \\
\domainT{\bk}{\bbh} & \defeq 
\begin{cases}
\{ i_{(j)} \in \cU: j \leq \bar{k} \} \cup \{ \bot\} ,  \text{ if } h_{i_{(\bk)}} > 0 
\\
\{ i_{(j)} \in \cU : j \leq p  \} \cup \{\top_{1}, \cdots \top_{\bk - p} \} \cup \{ \bot\} , \\
\qquad \qquad \qquad \qquad  \text { if } h_{i_{(p)}}> h_{i_{(p+1)}} = 0
\end{cases}
\end{align*}
Note that the labels $\{\top_j\}$ in $\bbh$ do not exist, and so for any index $\top_j$, its count is $h_{\top_j} = 0$. Consider the Gaussian mechanism $\gaussMechBot{\bk}(b; \bbh^{(0)}, \bbh^{(1)}, \tau )$ that takes a bit $b \in \{0,1\}$ and two neighboring histograms $\bbh^{(0)}$ and $\bbh^{(1)}$ with noise added to the top-$(\bk+1)$ elements from each histogram.  Because the labels need not be the same in the top-$(\bk+1)$ in $\bbh^{(0)}$ and $\bbh^{(1)}$, we assign a common label to the differing \emph{bad} indices, denoted as $\{B_\ell : \ell =1, \cdots, |\domainT{\bk}{\bbh^{(0)}} \setminus \domainT{\bk}{\bbh^{(1)}} | \}$.  

\begin{algorithm}[h!]
	\caption{$\gaussMechBot{\bk}$; Gaussian Mechanism over Limited Domain}
	\begin{algorithmic} 
		\State \textbf{Input:} Bit $b \in \{0,1\}$, neighboring histograms $\bbh^{(0)}$ and $\bbh^{(1)}$, cut off $\bk$, and parameter $\tau$.
		\State \textbf{Output:} Histogram $\bbv$ with labels in $\domainT{\bk}{\bbh^{(0)}} \cap \domainT{\bk}{\bbh^{(1)}}$ and $\{B_\ell : \ell \in [|\domainT{\bk}{\bbh^{(b)}} \setminus \domainT{\bk}{\bbh^{(1-b)}}|] \}$
		\State We relabel the indices in both $\bbh^{(0)}$ and $\bbh^{(1)}$ to form the following histogram $\bbv^{(b)}$ 
		\For{$i \in \domainT{\bk}{\bbh^{(0)}} \cap \domainT{\bk}{\bbh^{(1)}}$}
			\State $\bbv^{(b)} \gets\bbv^{(b)} \cup \{(h_i^{(b)}, i)\}$, where $h_{\top_j}^{(b)} = 0$. \algorithmiccomment{Keep common labels}
		\EndFor
		\State Initialize $\ell = 1$
		\For{$j \in \domainT{\bk}{\bbh^{(b)}} \setminus \domainT{\bk}{\bbh^{(1-b)}}$}\algorithmiccomment{Create ``bad" labels for uncommon elements} 
			\State $\bbv^{(b)} \gets \bbv^{(b)} \cup  \{(h_j^{(b)}, B_\ell)\}$.
			\State $\ell = \ell + 1$
		\EndFor
		\State $\bbv^{(b)} \gets \bbv^{(b)} \cup  \{(h_{(\bk + 1)}^{(b)} + 1 + \sqrt{2}\tau \Phi^{-1}(1 - \delta), \bot)\}$
		\State Add $\Normal{0}{\tau^2}$ to each count in $\bbv^{(b)}$ to form the noisy histogram $\hat\bbv$ \algorithmiccomment{Apply Gaussian Mechanism}
		\State Return $\hat\bbv$		
	\end{algorithmic}\label{algo:GaussMech}
\end{algorithm}

\begin{figure}
\centering
\includegraphics[width=0.75\textwidth]{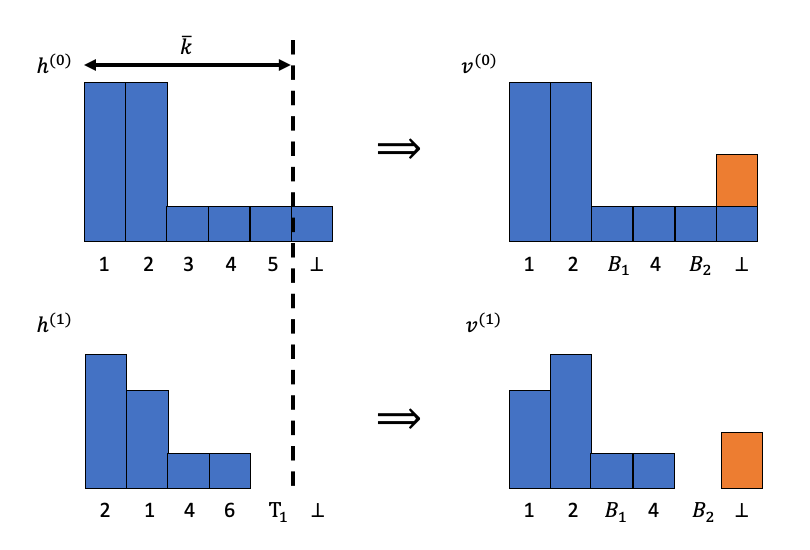}
\caption{Visualizing the construction of $v^{(b)}$ in $\gaussMechBot{\bk}(b;\bbh^{(0)}, \bbh^{(1)})$. Since labels $3,5$ are present in $h^{(0)}$ but not $h^{(1)}$ they get relabeled to $B_1,B_2$ respectively in $v^{(0)}$. Similarly, since the labels $6,\bot_1$ are present in $h^{(1)}$ but not in $h^{(0)}$ they get relabeled to $B_1,B_2$ in $v^{(1)}$.  The additional amount added to $\bot$ in orange denotes the added amount to the count $h_\bot$.}
\end{figure}

Note that once we fix neighboring histograms, $\gaussMechBot{\bk}(\cdot; \bbh^{(0)}, \bbh^{(1)}, \tau )$ is simply the Gaussian mechanism on a new  histogram $\bbv^{(b)}$ that uses the counts from $\bbh^{(b)}$ but whose labels include the common labels from $\bbh^{(0)}$ and $\bbh^{(1)}$, including the $\bot$ element and dummy elements $\{ \top_j \}$, as well as the bad indices $\{ B_j\}$.  Hence, we want to show that $\bbv^{(0)}$ and $\bbv^{(1)}$ can differ in at most $\Delta_0$ bins, i.e. the $\ell_0$-sensitivity of $\bbh^{(b)}$, and in any bin that changes, the counts can differ by at most $1$, i.e. the $\ell_\infty$-sensitivity of $\bbh^{(b)}$.  
We know that for any $j \in \domainT{\bk}{\bbh^{(0)}} \cap \domainT{\bk}{\bbh^{(1)}}$ that $|h_j^{(0)} - h_j^{(1)}| \leq 1$ and hence $|v_j^{(0)} - v_{j}^{(1)}|\leq 1$. We now consider the differing labels.
\begin{lemma}{\label{lemma:unkGauss_ell_infinity}}
Let $\bbh^{(0)}$ and $\bbh^{(1)}$ be neighbors with $\ell_\infty$-sensitivity $1$.  For any $i \in \domainT{\bk}{\bbh^{(b)}} \setminus \domainT{\bk}{\bbh^{(1-b)}}$ and $j \in \domainT{\bk}{\bbh^{(1-b)}} \setminus \domainT{\bk}{\bbh^{(b)}}$ we have $|h_i^{(b)} - h_j^{(1-b)}| \leq 1$.  Furthermore, $|h_{(\bk + 1)}^{(b)} - h_{(\bk + 1)}^{(1-b)}| \leq 1$.
\label{lem:v_ellinfty}
\end{lemma}
\begin{proof}
Without loss of generality, we assume that $\bbh^{(0)}$ has larger counts than $\bbh^{(1)}$.  If $i \in \domainT{\bk}{\bbh^{(0)}} \setminus \domainT{\bk}{\bbh^{(1)}}$ then we know that $h_i^{(0)} \geq h_{(\bk)}^{(0)}$ but $h_i^{(1)} \leq h_{(\bk)}^{(1)}$.  We also know that $h_i^{(0)} \leq h_{i}^{(1)} + 1$.  Putting this together, we have
\[
h_{(\bk)}^{(0)} \leq h_i^{(0)} \leq h_{i}^{(1)} + 1 \leq h_{(\bk)}^{(1)} + 1 .
\]
Similarly, for $j \in \domainT{\bk}{\bbh^{(1)}} \setminus \domainT{\bk}{\bbh^{(0)}}$ we have $h_{j}^{(0)}\leq h^{(0)}_{(\bk)}$ and $h_j^{(1)} \geq h_{(\bk)}^{(1)}$.  Further, $h_j^{(1)} \leq h_j^{(0)}$, which gives us
\[
h_{(\bk)}^{(1)} \leq h_j^{(1)}  \leq h_j^{(0)} \leq h^{(0)}_{(\bk)}.
\]
Combining the two, we have
\[
h_{(\bk)}^{(1)} \leq h_j^{(1)} \leq h_i^{(0)} \leq h_{(\bk)}^{(1)} + 1  \implies h_i^{(0)} - h_j^{(1)} \leq 1.
\]
Lastly, we have $h_{(\bk + 1)}^{(1)} \leq h_{(\bk + 1)}^{(0)}$.  Now assume that $h_{(\bk + 1)}^{(0)} > h_{(\bk + 1)}^{(1)} + 1$.  This can only occur if the $(\bk+1)$-th ranked element in $\bbh^{(0)}$ is not the same as the $(\bk + 1)$-th ranked element in $\bbh^{(1)}$, otherwise their count would differ by at most $1$.  Hence, there must be some element $i$ with count $h_{i}^{(1)} \leq h_{(\bk + 1)}^{(1)}$, but $h_{i}^{(0)} \geq h_{(\bk + 1)}^{(0)}$, since that would change the label for the $(\bk + 1)$-th ranked element between $\bbh^{(1)}$ and $\bbh^{(0)}$.  However, $h_i^{(0)} \leq h_i^{(1)} + 1$ and thus
\[
h_{(\bk + 1)}^{(0)} \leq h_i^{(0)} \leq  h_i^{(1)} + 1 \leq h_{(\bk + 1)}^{(1)} + 1.
\]
\end{proof}

We now show that the $\ell_0$-sensitivity between $\bbv^{(0)}$ and $\bbv^{(1)}$ is the same as between $\bbh^{(0)}$ and $\bbh^{(1)}$.
\begin{lemma}\label{lemma:unkGauss_ell_zero}
If $\bbh^{(0)}$ and $\bbh^{(1)}$ differ in at most $\Delta_0$ bins, then $\bbv^{(0)}$ and $\bbv^{(1)}$ also differ in at most $\Delta_0$ bins.
\label{lem:v_ell0}
\end{lemma}
\begin{proof}
Let $\ell$ be the number of bins that differ between $\bbh^{(0)}$ and $\bbh^{(1)}$ on labels in $\domainT{\bk}{\bbh^{(0)}} \cap \domainT{\bk}{\bbh^{(1)}}$.  Without loss of generality, we assume that $\bbh^{(0)}$ has larger counts than $\bbh^{(1)}$.  We know by definition that $\ell \leq \Delta_0$.  We now show that 
\[
|\domainT{\bk}{\bbh^{(b)}} \setminus \domainT{\bk}{\bbh^{(1-b)}}| \leq \Delta_0 - \ell.
\]
It suffices to only consider $|\domainT{\bk}{\bbh^{(0)}} \setminus \domainT{\bk}{\bbh^{(1)}}|$ since $\domainT{\bk}{\bbh^{(0)}} \setminus \domainT{\bk}{\bbh^{(1)}}$ has the same cardinality.  Note that for any $i \in \domainT{\bk}{\bbh^{(0)}} \setminus \domainT{\bk}{\bbh^{(1)}}$, that implies $h_i^{(0)} > h_{i}^{(1)}$, and we know only $\Delta_0 -\ell$ such additional indices can exist.
If $h_i^{(0)} = h_i^{(1)}$, then the position of index $i$ cannot have moved up the ordering from $\bbh^{(1)}$ to $\bbh^{(0)}$ because we assumed $\bbh^{(0)}$ had larger counts.  Therefore, if $i \notin \domainT{\bk}{\bbh^{(1)}}$ and $h_i^{(0)} = h_i^{(1)}$ we must also have $i \notin \domainT{\bk}{\bbh^{(0)}}$.  Hence, $|\domainT{\bk}{\bbh^{(0)}}\setminus \domainT{\bk}{\bbh^{(1)}}| \leq \Delta_0 - \ell$
\end{proof}

With these two results we can show the following.
\begin{lemma}\label{lem:GaussMechBot}
For any two neighboring histograms $\bbh^{(0)}$ and $\bbh^{(1)}$ with $\ell_0$-sensitivity $\Delta_0$ and $\ell_\infty$-sensitivity $1$, the procedure $\gaussMechBot{\bk}(\cdot; \bbh^{(0)}, \bbh^{(1)},\tau)$ is $\tfrac{\Delta_0}{2\tau^2}$-zCDP
\end{lemma}
\begin{proof}
Follows directly from the sensitivity analysis of the intermediate histogram $\bbv^{(b)}$ from Lemmas~\ref{lem:v_ellinfty} and~\ref{lem:v_ell0}
\end{proof}

We now show that, for a pair of fixed neighboring datasets, $\unkGauss{\bk}$ is equivalent to running a post processing function on $\gaussMechBot{\bk}$ for certain outcomes and that we can bound the probability of other outcomes where they do not align.  We now define \emph{good} and \emph{bad} outcome sets.
\begin{definition}
Given two neighboring histograms $\bbh^{(0)},\bbh^{(1)}$, we define $\cS_{\gauss}^{(b)}$ as the outcome set of $\unkGauss{\bk}(\bbh^{(b)}; \tau , \delta)$.
We then define the \emph{bad outcomes} as 
$
\cB^{(b)}_{\gauss} \defeq \cS_{\gauss}^{(b)} \setminus \cS_{\gauss}^{(1-b)}, \text{ for } b \in \{ 0,1\}.
$
\label{defn:eps_delt_setsGauss}
\end{definition}
Next, we bound the probability of outputting something in $\cB^{(b)}_{\gauss}$, and also show that we can achieve pure DP for the remaining outputs that are common in $\bbh^{(0)}$ and $\bbh^{(1)}$. For bounding the bad outcomes, it suffices to consider each element in $\domainT{\bk}{\bbh^{(b)}} \setminus \domainT{\bk}{\bbh^{(1-b)}}$ and bound the probability that its respective noisy value is above a threshold $h_\bot$ with added noise. Note that the threshold computation will have a simpler analysis than prior work due to the sum of two Gaussians being Gaussian, whereas \citet{DurfeeRo19} considered Laplace noise which does not satisfy the same property.

\begin{lemma}\label{lem:gauss_del}
For neighboring histograms $\bbh^{(0)},\bbh^{(1)}$ with $\ell_0$-sensitivity $\Delta_0$ and $\ell_\infty$-sensitivity $1$, with $b \in \{0,1\}$ we have
\begin{equation}
\Pr[\unkGauss{\bk}(\bbh^{(b)}; \tau,\delta) \in \cB^{(b)}_{\gauss}] \leq  \delta \Delta_0.
\label{eq:barDelta}
\end{equation}
\end{lemma}

In order to prove Lemma~\ref{lem:gauss_del}, we will need to define a mechanism that takes an input domain of indices, as well as a histogram.
\begin{definition}\label{defn:rnmk}[Sorted Gaussian Mechanism over Limited Domain]
We define the \emph{sorted Gaussian mechanism over limited domain} to be $\gaussMax{\bk}$ that takes as input a histogram along with a domain set of indices $\cD$ and returns an ordered list of elements until $\bot$'s count, that is

\begin{equation*}
\gaussMax{\bk}(\bbh,\cD) =
\{(v_{i_{(1)}}, i_{(1)}),...,(v_{i_{(1)}}, i_{(j)}),(v_{\bot},\bot) \}
\end{equation*}
where $(v_{i_{(1)}},...,v_{i_{(j)}},v_\bot)$ is the sorted list until $v_\bot$ of $v_i =  \Normal{h_{(i)}}{\tau^2}$ and $v_\bot = \Normal{h_{\bot}}{\tau^2}$, for each $i \in \cD$ and 
\begin{equation}
h_\bot \defeq h_{(\bk + 1)} + 1 +\sqrt{2} \tau \Phi^{-1}(1 - \delta)
\label{eq:h_bot}
\end{equation}

\end{definition}
Note that $\gaussMax{\bk}(\bbh,\cD^{\bk}(\bbh))$ and $\unkGauss{\bk}(\bbh)$ are equal in distribution.
We will use the following result to prove Lemma~\ref{lem:gauss_del}.
\begin{lemma}\label{lem:delta_bound}
Given an histogram $\bbh$ and some domain $\cD$ that can include dummy $\{ \top_i\}$. For any $i \in \cD$ such that $h_i \leq h_{(\bk + 1)} + 1$, then 
\[
\Pr[i \in \gaussMax{\bk}(\bbh,\cD)] \leq \delta.
\]
\end{lemma}
\begin{proof}
For simplicity, we will set $T=\tau\sqrt{2} \Phi^{-1}(1 - \delta)$, which implies $h_{\bot} = h_{(\bk + 1)} +1 +  T$ and plug back in at the end of the analysis.
By construction of our mechanism, we know that the noisy estimate of $h_i$ must be greater than the noisy estimate of our threshold $h_{\bot} = h_{(\bk + 1)} + 1 +  T$ to be a possible output, which implies

\[
\Pr[i \in \gaussMax{\bk}(\bbh,\cD)] \leq \Pr[h_i + \Normal{0}{\tau^2} > h_\bot + \Normal{0}{\tau^2}]. 
\]

By assumption, $h_i \leq h_{(\bk + 1)} + 1$, and by the fact that the sum of two independent Gaussians $Z_1, Z_2 \sim \Normal{0}{\tau^2}$ is also Gaussian, i.e.  $Z_1 + Z_2 \sim \Normal{0}{2\tau^2}$,
$$
\Pr[i \in \gaussMax{\bk}(\bbh,\cD)]  \leq \Pr[\Normal{0}{2\tau^2} >  T ] = 1 - \Phi\left( \frac{T}{\sqrt{2} \tau} \right).
$$
Plugging in $T = \sqrt{2} \sigma \Phi^{-1}(1 - \delta)$ gives the result.

\end{proof}

We can now prove Lemma~\ref{lem:gauss_del}.

\begin{proof}[Proof of Lemma~\ref{lem:gauss_del}]

This will follow from a simple union bound on each  $i \in  \domainT{\bk}{\bbh^{(0)}}\setminus \domainT{\bk}{\bbh^{(1)}}$ where we consider each subset of $\cB^{(0)}_{\gauss}$ such that each outcome contains $i$, or more formally we define $\cB^{(0)}_{\gauss}(i) \defeq \{o \in \cB^{(0)}_{\gauss}: i \in o\}$
This then implies that 

\[
\Pr[\gaussMax{\bk}(\bbh^{(0)},\domainT{\bk}{\bbh}) \in \cB^{(0)}_{\gauss}] \leq \sum_{i \in \domain{\bk}{\bbh^{(0)}}\setminus  \domain{\bk}{\bbh^{(1)}}} \Pr[\gaussMax{\bk}(\bbh,\domainT{\bk}{\bbh}) \in \cB^{(0)}_{\gauss}(i)] 
\]
because each outcome $o \in \cB^{(0)}_{\gauss}$ must contain some $i \in  \domainT{\bk}{\bbh^{(0)}}\setminus \domainT{\bk}{\bbh^{(1)}}$ by construction. Furthermore, by construction we also have 

\[
\Pr[\gaussMax{\bk}(\bbh^{(0)},\domain{\bk}{\bbh^{(0)}}) \in \cB^{(0)}_{\gauss}(i)] = \Pr[i \in \gaussMax{\bk}(\bbh^{(0)},\domain{\bk}{\bbh^{(0)}})]
\]

Our claim then immediately follows from Lemma~\ref{lem:delta_bound} and the fact that the size of $\domain{\bk}{\bbh^{(0)}}\setminus  \domain{\bk}{\bbh^{(1)}}$ is at most $\Delta_0$ by Lemma~\ref{lem:v_ell0}.

\end{proof}

We can now prove Theorem~\ref{thm:main_UnkGauss}.
\begin{proof}[Proof of Theorem~\ref{thm:main_UnkGauss}]
We will use Lemma~\ref{lem:meta} to prove this result.  From Lemma~\ref{lem:gauss_del}, we have the probability of bad outcomes being negligible.  We now need to define a mechanism $A(b;\bbh^{(0)}, \bbh^{(1)})$ that matches $\unkGauss{\bk}$ on good outcomes and is DP.  Lemma~\ref{lem:GaussMechBot} shows that $\gaussMechBot{\bk}(\cdot; \bbh^{(0)}, \bbh^{(1)})$ is DP.   We then define a post processing function on $\gaussMechBot{\bk}$.  First, we sort in descending order the elements up until we hit $\bot$ and then we eliminate the rest.  Next, we drop all the dummy labels $\{\top_i\}$ and their noisy counts.  We know from Lemma~\ref{lem:dummy} that sorting up to $\bot$ and dropping the dummy labels is equivalent to never considering the dummy elements in the first place.  Note that this post processing function on $\gaussMechBot{\bk}(b;\bbh^{(0)}, \bbh^{(1)})$ is equivalent to our main algorithm $\unkGauss{\bk}$ for good outcomes.  Because post-processing cannot increase the privacy loss parameters, we can use Lemma~\ref{lem:meta} with $A(b; \bbh^{(0)}, \bbh^{(1)})$ as this post-processing function of $\gaussMechBot{\bk}(b;\bbh^{(0)}, \bbh^{(1)})$.
\end{proof}

%% file: unkDomain3.tex
\section{Restricted $\ell_0$-sensitivity, Unknown Domain \label{sect:UnkBase}} 

We turn back to the continual observation setting where a user can contribute at most $\Delta_0$ many items at any round, but the domain of items is unknown.  When the domain is not given in advance, it is impossible for an item that no one contributed to in the stream to actually be returned.  However, it is important to point out that the mere existence of a particular item shows that someone in the dataset must have contributed such an item. We will then impose a threshold so that the probability that we display an item with a single count is very small.  We emphasize that even if the domain were known in advance, it still might be desirable to consider this setting, since the domain might be incredibly large making $\knownBase$ computationally expensive.

\begin{algorithm}[h!]
	\caption{$\unkBase$; Return a running histogram}
	\begin{algorithmic}
		\State  \textbf{Input:} Stream $\omega_{1:T} = \omega_1, \cdots, \omega_T$, base $r$, noise $\tau$, $\delta$.
		\State  \textbf{Output:} Noisy histograms $\hat{\bbh}_{1:T}$.
		\For{$t\in [T]$}
			\State  $\cD_t = \emptyset, \cZ = \emptyset$, 
			Let $\cU_t$ be the set of items in $\omega_{1:t}$
			\For{ $u \in \cU_t$}
			\State  Let $\cI_t(r)$ be set of cells $(j,\ell)$ that are used in the representation of $t$ with base $r$ 
				\For{each cell $(j,\ell) \in \cI_t(r)$} \State \algorithmiccomment{Add noise to each $u$, if noise has been added to $u$ before, add the same realization}
					\If{ $ Z_{j,\ell}^u \notin \cZ$}
						\State  Let $Z_{j,\ell}^u \sim \Normal{0}{L_r \tau^2}$, 
						$\cZ \leftarrow \cZ \cup \{(u,j,\ell)\}$
					\EndIf
				\EndFor
				\State  $\hat{h}_t^u = h_t^u + \sum_{(j,\ell) \in \cI_t(r)} Z_{j,\ell}^u$
				\If{$\hat{h}_t^u > m_\delta$ from \eqref{eq:threshold}} \algorithmiccomment{Only add elements above the threshold}
					\State  $\cD_t \gets \cD_t \cup \{u \}$
				\EndIf
			\EndFor
		\EndFor
		\State  Return $( \{ (\hat{h}_{t}^u, u) : u \in \cD_t \} : t \in [T])$.	
	\end{algorithmic}\label{algo:UnkBase}
\end{algorithm}

We present the main algorithm of this section in Algorithm~\ref{algo:UnkBase}.\footnote{The sets $\cI_t(r)$ can be built using the same idea as in $\BinMech$, all one needs is the $r$-nary representation of $t$. For clarity of exposition we do not build the sets in the pseudocode of $\unkBase$.}  We can summarize $\unkBase{}$ as simply taking the items, denoted as $\cU_t$, that have appeared in the stream up to time $t$, form their current histogram $\bbh\left(\omega_{1:t}; \cU_t\right)$, add noise in the way one would in $\knownBase{}$, but only show items if their noisy count is above the following threshold,
\begin{align}\label{eq:threshold} 
m_\delta & \defeq \tau L_r \sqrt{r-1} \Phi^{-1}\left( 1 - \delta / T \right) + 1.
\end{align}

We point out that the algorithm discovers a new set of items $\cD_t$ at each round $t \in [T]$, essentially wiping away the set of items that have already appeared at previous rounds.  However, we still ensure the same privacy level if the algorithm remembers previous items that were discovered but may not have noisy count above the threshold at a later round. This would avoid the strange behavior of some items having a count in some rounds and then disappearing in other rounds, however one would need to remember all the items that were previously discovered at each round.

\subsection{Privacy Analysis}

Our analysis of $\unkBase$ can be thought of as a generalization of the stability based histograms studied in earlier work \citep{KorolovaKeMiNt09, BunNiSt16, Vadhan17}, where only elements with positive counts exist in the histogram.  Directly applying the stability based histogram approach would result in another variant of the $\texttt{MetaAlgo}$, as the histogram in each cell in the partial sum table would have counts below a certain threshold removed.  
We opted to using $\texttt{UnkGauss}$ in the presentation of the $\texttt{MetaAlgo}$ because it is more general than the original stability based histogram approaches, due to it only needing access to the top-$\bar{k}$ counts with positive counts, as opposed to all positive counts.  

The novelty of our approach is then in extending the stability based histogram approach to the case where we only have access to positive counts up to and including round $t \in [T]$, rather than in all sub-streams.  Lemma~\ref{lem:meta} allows us to consider a DP algorithm with access to a given pair of neighboring datasets $x, x'$ and only consider outcomes that can occur with both neighbors, i.e. \emph{good} outcomes.  Hence, we can then consider the two partial sum tables that suffices to compute the running counts for either $x$ or $x'$.  The problem between the two partial sum tables is that there are table cells with elements present for $x$ but not for $x'$.  To address this issue, we introduce zero count elements to each cell, so that each cell has the same number of elements that get noise added to it.  Note that the labels of the zero counts need to be made common across the two partial sum tables, which we can do because we are constructing a DP algorithm that knows $x$ and $x'$.  We can then analyze the privacy of this resulting partial sum table using composition of Gaussian mechanisms.  The last part in our analysis is to bound the probability of all \emph{bad} outcomes, which in this case is when any of the elements that had zero count in $x'$ yet positive count in $x$ appear in any histogram in any $t \in [T]$, which we can do by applying a threshold that is determined by the tail bound of the sum of at most $L_r$ Gaussians, which itself is Gaussian (another reason to use Gaussian noise!).  

We now present the full analysis.  As we did for $\unkGauss{\bk}$, we will instead analyze a slight variant of $\unkBase$, which we call $\unkBaseTop{\bar{d}}$, see Algorithm~\ref{algo:UnkBaseTop}. $\unkBaseTop{\bar{d}}$ pads the set of items with dummy items $\{\top_i\}$ that we add noise to in each cell of the partial sum table, so that each cell has the same cardinality $\bar{d}$ of items, which is some upper bound on the dimension of the set of items.  
\begin{algorithm}[h!]
	\caption{$\unkBaseTop{\bar{d}}$; Return a running histogram}
	\begin{algorithmic}
		\State  \textbf{Input:} Same as $\unkBase$ and an upper bound $\bar{d}$.
		\State  \textbf{Output:} Noisy histograms $\hat{\bbh}_{1:T}$.
		\State  Use threshold $m$ from \eqref{eq:threshold}. 
		\For{cell $(j,\ell)$ in partial histogram table}
				\State \algorithmiccomment{Create the noisy partial histogram table, pad if necessary }
			\State  Define $\cU_{j,\ell}$ as the set of items in the corresponding substream of $\omega_{1:T}$ and let $d_{j,\ell} = |\cU_{j, \ell}|$.
			\State  Include dummy items $\top_{j,\ell} = \{\top_{j,\ell}^1, \cdots \top_{j,\ell}^{\bar{d} - d_{j,\ell}} \}$, let $n_{j,\ell} = \vert \top_{j,\ell} \vert$.
			\State  Form the partial histogram $p_{j,\ell} = (p_{j,\ell}^u: u \in \cU_{j,\ell} \cup \top_{j,\ell})$ for this cell.
			\State  Add independent noise to each count in this cell to get $\hat{p}_{j,\ell} = (p_{j,\ell}^u +\Normal{0}{L_r\tau^2} : u \in \cU_{j,\ell} \cup \top_{j,\ell})$.
			\State  We then have histogram with labels $\hat{\bbp}_{j,\ell} = \{ (\hat{p}_{j,\ell}^u, u) : u \in \cU_{j,\ell} \cup \top_{j,\ell}  \}$
		\EndFor
		\For{$t\in [T]$}
			\State  $\cD_t = \emptyset$
			\State  Let $\cI_t(r)$ be set of cells $(j,\ell)$ that are used in the representation of $t$ with base $r$.
			\State  Let $\cU_t $ be the union of items present in each cell used for the count at time $t$
			\State  Let $\top_t$ be the union of dummy items present in each cell.
			\For{each $(j,\ell)$ cell in $\cI_t(r)$ }
			\State \algorithmiccomment{Relabel dummy items in each cell with items that have newly appeared in the stream}
				\For{$u \in \cU_t \setminus \cU_{j,\ell}$}
					\State  Replace the dummy label with largest index $n_{j,\ell}$ to $u$, i.e. $(\hat{p}_{j,\ell}^{u}, u) \gets (\hat{p}_{j,\ell}^{\top_{j,\ell}^{n_{j,\ell}}}, \top_{j,\ell}^{n_{j,\ell}})$.
					\State  Update $\cU_{j,\ell} \gets \cU_{j,\ell} \cup \{ u\}$, $\top_{j,\ell} \gets \top_{j,\ell} \setminus \{\top^{n_{j,\ell}}_{j,\ell} \}$, and $n_{j,\ell} \gets n_{j,\ell} - 1$.
 				\EndFor
			\EndFor
			\For{ $u \in \cU_t \cup \top_t$} \algorithmiccomment{Aggregate histograms from each cell $\cI_t(r)$}
				\State  $\hat{h}_t^u = \sum_{(j,\ell) \in \cI_t(r)}  \hat{p}_{j,\ell}^u$
				\If{$\hat{h}_t^u > m_\delta$} \algorithmiccomment{Only add elements above the threshold }
					\State  $\cD_t \gets \cD_t \cup \{u \}$
				\EndIf
			\EndFor
		\EndFor
		\State  Return $( \{ (\hat{h}_{t}^u, u) : u \in \cD_t \} : t \in [T])$.	
	\end{algorithmic}\label{algo:UnkBaseTop}
\end{algorithm}
Similar to Lemma~\ref{lem:dummy}, we will show that simply dropping these dummy items later is equivalent to having never considered them.  

\begin{lemma}\label{lem:dummy2}
Let $h \in \N^{p}$ be a histogram with labels $\{i_1, \cdots, i_{p}\}$. Let $M(h)$ be the following for $\bar{d}\geq p$ and $Z_j \stackrel{i.i.d.}{\sim} \cP_j$ and $\hat{h}_j = h_j + Z_j$,
\[
\{ (\hat{h}_1, i_1), \cdots, (\hat{h}_p, i_{p} ), (Z_{p+1}, \top_1),  \cdots, (Z_{\bar{d}},\top_{\bar{d} - p} ) \}.
\]
Let $M'(h)$ be the mechanism that drops all items with counts lower than some threshold $m$ and drops any item with label in $\{\top_i\}$.  Now let $\hat{M}(h)$ be the same as $M(h)$ except it does not include the $\{\top_i\}$ items.  Then $M'(h)$ is equal in distribution to $\hat{M}'(h)$.
\end{lemma}
\begin{proof}
We need to show that adding independent noise to $\bar{d}$ counts, of which some have $\{\top_i\}$ labels and then dropping these terms is equivalent to having never considered those items.

Let $f(\cdot)$ be the density function for distribution $\cP$, $f_{M'}(\cdot)$ be the density of $\hat{M}$, and $f_{\hat{M}'}$ be the density of $\hat{M}'$.  We fix an outcome of counts  $(z_1, z_2, \cdots, z_k)$ with $k \leq p$ and denote the set of indices that are not in this outcome to be $I$ after dropping counts of $\{ \top_i \}$.  We then have the density for mechanism $\hat{M}'$ as 
\begin{align*}
f_{\hat{M}'}(z_1, \cdots, z_k) & =  \prod_{i=1}^k f_i(z_i - h_i) \int_{\infty}^{m} \cdots \int_{\infty}^{m}  \prod_{\ell \in I } f_\ell(z_\ell - h_\ell)  dz_\ell     \\
& =   \prod_{i=1}^k f_i(z_i - h_i) \int_{\infty}^{m} \cdots \int_{\infty}^{m}  \prod_{\ell \in I } f_\ell(z_\ell - h_\ell)  dz_\ell   \cdot \prod_{j = p+1}^{\bk} \int_{\R}  f(z_j) dz_{j} \\
& =  f_{M'}(z_1, \cdots, z_k).
\end{align*}
\end{proof}

Algorithm~\ref{algo:GaussMechBase} is a variant of the Gaussian Mechanism that we use when given two neighboring histograms.  Note its similarity with Algorithm~\ref{algo:GaussMech}.  The algorithm $\gaussMech{\bar{d}}$ takes a parameter $\bar{d}$ which is an upper bound on the number of distinct bins of the histograms, this ensures that each cell has access to a full histogram. We will assume that we have access to the full histogram, including the items with 0 counts, rather than just having the top-$(\bk+1)$ as it was assumed in $\gaussMechBot{\bk}$.  

\begin{algorithm}[h!]
	\caption{$\gaussMech{\bar{d}}$; Gaussian Mechanism over Full Domain}
	\begin{algorithmic} 
		\State  \textbf{Input:} Bit $b \in \{0,1\}$, neighboring histograms $\bbh^{(0)}$ and $\bbh^{(1)}$, upper bound  $\bar{d}$, and parameter $\tau$.
		\State  \textbf{Output:} Histogram $\hat{\bbv}$ with labels in $\domainT{\bar{d}}{\bbh^{(0)}} \cap \domainT{\bar{d}}{\bbh^{(1)}}$ and $\{B_\ell : \ell \in [|\domainT{\bar{d}}{\bbh^{(b)}} \setminus \domainT{\bar{d}}{\bbh^{(1-b)}}|] \}$
		\State  Let $\bbv^{(b)} = \emptyset$
		\State  We relabel the labels in both $\bbh^{(0)}$ and $\bbh^{(1)}$ to form the following histogram $\bbv^{(b)}$ 
		\For{$i_{(j)} \in \domainT{\bar{d}}{\bbh^{(0)}} \cap \domainT{\bar{d}}{\bbh^{(1)}}$} \algorithmiccomment{Add common labels}
			\State  $\bbv^{(b)} \gets\bbv^{(b)} \cup \{(h_{i_{(j)}}^{(b)}, i_{(j)})\}$, where $h_{\top_j}^{(b)} = 0$.
		\EndFor
		\State  Initialize set of labels that have a non-dummy label and have a different label in the two datasets
		\[
		\cB = \left\{  \left\{\domainT{\bar{d}}{\bbh^{(1)}} \setminus \domainT{\bar{d}}{\bbh^{(0)}}  \right\} \cup \left\{ \domainT{\bar{d}}{\bbh^{(0)}} \setminus \domainT{\bar{d}}{\bbh^{(1)}}\right\}\right\} \setminus \{\top_i  : i \in [\bar{d}] \}
		\]
		\For{$j \in \domainT{\bar{d}}{\bbh^{(b)}} \setminus \domainT{\bar{d}}{\bbh^{(1-b)}}$} \algorithmiccomment{Add uncommon labels}
			\If{$h_j^{(b)} \geq h_j^{(1-b)}$}
				\State  $\bbv^{(b)} \gets \bbv^{(b)} \cup  \{( h_j^{(b)}, j)\}$.
			\Else
				\State  Select a label from $\cB$, call it $a$
				\State  $\bbv^{(b)} \gets \bbv^{(b)} \cup  \{(0,a)\}$.
				\State  $\cB \gets \cB \setminus \{a\}$
			\EndIf
		\EndFor
		\State  Add $\Normal{0}{\tau^2}$ to each count in $\bbv^{(b)}$ to form the noisy histogram $\hat\bbv$ \algorithmiccomment{Gaussian Mechanism}
		\State  Return $\hat\bbv$		
	\end{algorithmic}\label{algo:GaussMechBase}
\end{algorithm}

We then have the following privacy guarantee of $\gaussMech{\bar{d}}$.

\begin{lemma}\label{lem:GaussMechBase}
For any two neighboring histograms $\bbh^{(0)}$ and $\bbh^{(1)}$ with $\ell_0$-sensitivity $\Delta_0$ and $\ell_\infty$-sensitivity $1$, the procedure $\gaussMech{\bar{d}}(\cdot; \bbh^{(0)}, \bbh^{(1)}, \tau)$ is $\frac{\Delta_0}{2\tau^2}$-zCDP.
\end{lemma}
\begin{proof}
Follows the same analysis as in Lemmas~\ref{lem:v_ellinfty} and~\ref{lem:v_ell0}
\end{proof}

We now show that we can connect $\gaussMech{\bar{d}}$ with our $\unkBaseTop{\bar{d}}$ algorithm on good outcomes, which brings us a step closer to being able to use Lemma~\ref{lem:meta}.  
\begin{lemma}\label{lem:equalAMechanism}
For neighbors $\omega^{(0)}_{1:T}$ and $\omega^{(1)}_{1:T}$ and outcomes $\cG_{\unk}$ that can occur in both $\unkBaseTop{\bar{d}}(\omega_{1:T}^{(b)}; \tau)$ for $b \in \{ 0,1\}$, there exists a mechanism $A(\cdot;\omega_{1:T}^{(0)}, \omega_{1:T}^{(1)}) $ that is $\frac{\Delta_0}{2\tau^2}$-zCDP and for any outcome set $S \subseteq \cG_{\unk}$ we have
$
\Pr[\unkBaseTop{\bar{d}}(\omega^{(b)};\tau) \in S] = \Pr[ A(b; \bbh^{(0)}, \bbh^{(1)}) \in S].
$
\end{lemma}
\begin{proof}
In $\unkBaseTop{\bar{d}}(\bbh^{(b)}; \tau)$, we are essentially applying the Gaussian mechanism to a histogram in each cell of the partial histogram table.  Consider a cell $(i,j)$ that differs between streams $\omega_{1:T}^{(0)}$ and $\omega_{1:T}^{(1)}$, of which there can be as many as $L_r$ cells.  We apply $\gaussMech{\bar{d}}(b; \bbh^{(0)}, \bbh^{(1)}, L_r\tau)$ to this cell's histogram of counts.  Doing this across all cells that can actually change between the two neighboring streams, we can apply composition and Lemma~\ref{lem:GaussMechBase} to get that releasing the full partial sum table is $\tfrac{\Delta_0}{2\tau^2}$-zCDP.  

We then apply a post-processing function that adds up the corresponding cells of the table to get the aggregate count for each time $t \in [T]$ and removes any item that has count lower than $m$.  Because we are only considering good outcomes, this will ensure that any count with a bad label $\{B_i\}$ is not in the result, hence bad labels have noisy counts less than $m$.  Note that these bad labels are the only terms that could have had different labels than the counts returned in $\unkBaseTop{\bar{d}}(\omega_{1:T}^{(b)}; \tau)$.  Hence, the resulting mechanism is equivalent to $\unkBaseTop{\bar{d}}(\omega_{1:T}^{(b)}; \tau)$ for outcomes in $\cG_{\unk}$.
\end{proof}

We next need to figure out the right threshold $m$ to set that will ensure that bad outcomes occur with negligible probability.  The only way an item $u$ that occurred once in a stream $\omega_{1:T}$ but not in another neighboring stream $\omega'_{1:T}$ can be returned is if there is a noisy count $\hat{h}_t^u > m$ for some $t \in [T]$ and item $u$ that was present in $\omega$ but not in $\omega'$ or vice versa.  
Since the counts are computed as a function of the partial histogram table $\hat{\bbp}$, we need to make sure that all the noisy counts in this table for items that are not common in $\omega$ and $\omega'$ cannot add up to something larger than $m$.  

\begin{lemma}
\label{lem:badOutcomes}
Fix neighbors $\omega_{1:T}^{(0)}$ and $\omega_{1:T}^{(1)}$ with $\ell_0$-sensitivity $\Delta_0$ and define $\cB_{\unk}^{(b)}$ to be the set of outcomes that can occur in $\unkBaseTop{\bar{d}}(\omega_{1:T}^{(b)};\tau; \delta)$ but not in $\unkBaseTop{\bar{d}}(\omega_{1:T}^{(1-b)};\tau,\delta)$.  Then we have
$
\Pr[\unkBaseTop{\bar{d}}(\omega_{1:T}^{(b)};\tau,\delta) \in \cB_{\unk}^{(b)}] \leq \Delta_0\delta.
$
\end{lemma}
\begin{proof}
We first consider the probability that an item from $\cB_{\unk}^{(b)}$ can be returned at a given time $t$, which must mean that there is a dummy label $\top_t^i$ for one stream and a real label $a$ for the other stream at time $t$.  Let's consider the first time $t$ that the labels do not align in the neighboring streams.  The only way this could happen is if this were the first time $a$ appeared in the stream, since all prior items in both streams are the same.  Hence, the true count of $a$ at time $t$ will be 1 in one stream and 0 in the other.  The additional noise must have caused its count to appear above the threshold $m$.  We then compute the probability that a count can appear above threshold $m$.  This threshold then needs to be set so that all future times will also not have noisy count on $a$ above the threshold until someone else has item $a$ in the stream.  Further, there can be at most $\Delta_0$ many items like $a$, implying that all items that a user contributes at a round $t$ are all the first time they appeared in the stream.

There are multiple ways to do this.  One is to bound the probability that all of the independent Gaussians that are used to compute the count for item $a$ are below $m / \left( (r-1) L_r \right)$, hence any sum of at most $(r-1)L_r$ terms is below $m$.  Another way, is to bound the probability that any sum of these independent Gaussians is below $m$, and take a union bound over all $T$ rounds.  We opt for the latter approach.
\begin{align*}
& \Pr[\unkBaseTop{\bar{d}}(\omega_{1:T}^{(b)};\tau,\delta) \in \cB_{\unk}^{(b)}]  = \Pr[ \exists a \in \cB_{\unk}^{(b)} \text{ s.t. label } a \in \unkBaseTop{\bar{d}}(\omega_{1:T}^{(b)};\tau,\delta) ] \\
& \qquad \leq  \Pr_{\{ Z_{i,j}^u \} \stackrel{i.i.d.}{\sim} \Normal{0}{L_r\tau^2}}\left[ \max_{t \in [T], u \in \cB_{\unk}^{(b)}} \left\{ 1+  \sum_{(j,\ell) \in \cI_t(b)} Z_{i,j}^u \right\} > m_\delta\right] \\
& \qquad \leq \Delta_0 T \cdot \Pr\left[ \Normal{0}{(r-1)L_r^2 \tau^2} > m_\delta - 1\right] \\
& \qquad = \Delta_0 T \cdot \left( 1- \Phi\left(\frac{m_\delta-1}{L_r\tau \sqrt{r-1}}\right) \right).
\end{align*}
The last inequality follows from a union bound.  Setting $m_\delta$ as in \eqref{eq:threshold} gives the bound of $\Delta_0 \delta$.
\end{proof}

We can now state our privacy result, which follows from the privacy of $\unkBaseTop{\bar{d}}$ and recalling that $\unkBase{}$ is a post-processing function of $\unkBaseTop{\bar{d}}$.

\begin{theorem}
$\unkBase{}(\cdot; \tau,\delta)$ is $(\epsilon(\tfrac{\Delta_0}{2\tau^2}, \delta'),\Delta_0\delta+\delta')$-DP for any $\delta' >0$ with $\epsilon(\cdot, \cdot)$ given in \eqref{eq:epsilon}.
\end{theorem}
\begin{proof}
We first show that $\unkBaseTop{\bar{d}}$ is DP.  This follows by applying Lemma~\ref{lem:meta} with Lemmas~\ref{lem:equalAMechanism} and ~\ref{lem:badOutcomes}.  Now we use Lemma~\ref{lem:dummy2} to show that at any round $t$, dropping the dummy labels $\top_t$ is equivalent to never adding noise to them.  However, we may use the noise allocated for a dummy item in some cells at later rounds. In particular, we replace the dummy label when a new item appears at a later point in $\unkBaseTop{\bar{d}}$.  Whether this noise was drawn earlier for that cell or at the time that it is actually used, both give the same distribution.  Hence, the post-processing function of dropping dummy labels at each round of $\unkBaseTop{\bar{d}}(\omega_{1:T}; \tau,\delta)$ is equivalent to running $\unkBase{}(\omega_{1:T}; \tau,\delta)$.
\end{proof}

\subsection{Utility Analysis}
We then turn to analyzing the utility of $\unkBase{}$.  First we consider the probability that a particular item will appear in the result at time $t$.

\begin{lemma}
Let $\cD_t$ be the discovered set at round $t$ in $\unkBase{}(\omega_{1:T}; \tau,\delta)$.  Let $h_t^u = \sum_{\ell = 1}^t \1{u \in \omega_\ell}$ be the true count for item $u$ in the stream up to round $t$ and assume that it is larger than the threshold $h_t^u = m_\delta + c \cdot \tau$ for some $c > 0$.  We can then bound the probability that $u$ is part of the discovered set at time $t$, 
\[
\Pr[ u \in \cD_t] \geq \Phi\left( \frac{c}{\sqrt{r - 1} L_r} \right).
\]
\end{lemma}
\begin{proof}
We will write $\cI_r(t)$ as the set of indices in the partial histogram table that gets used to compute the counts at time $t$.  Recall that we will add noise $\Normal{0}{|\cI_r(t)| L_r\tau^2}$ to the true count $h_t^u = m_\delta + c\cdot \tau$ at time $t$.  We then need to ensure that the noisy count will be above the threshold $m_\delta$ given in \eqref{eq:threshold}.  Hence, we have
\begin{align*}
\Pr[ u \in \cD_t]  & = \Pr[\Normal{h_t^u}{|\cI_r(t)| L_r\tau^2} > m_\delta] = 1 - \Phi\left( \frac{-c}{\sqrt{|\cI_r(t)|L_r}} \right) \\
& = \Phi\left(\frac{c}{\sqrt{|\cI_r(t)| L_r}} \right) \geq \Phi\left( \frac{c}{\sqrt{r - 1} L_r} \right)
\end{align*}
\end{proof}
Additionally, for those items $u \in \cU$ that the algorithm releases, we provide bounds on the difference between their true count $h_t^u$ and their noisy count $\hat{h}_t^u$ by noticing that $\hat{h}_t^u$ is the original count $h_t^u$ plus Gaussian noise truncated at threshold $m_\delta$.
\begin{lemma}
Given a stream $\omega_{1:T}$ and an item $u \in \cD_t$ that is part of the discovered set from $(\hat{h}_t^u, u) \in \unkBase{}(\omega_{1:T};\tau,\delta)$ at time $t$, we can then bound the error on its true count $h_t^u= m_\delta + c \cdot \tau$ at time $t$ for $0< c<\eta$ with high probability,
\begin{align*}
& \Pr[ |h_t^u - \hat{h}_t^u| \geq \eta \tau \mid u \in \cD_t] \\
& \qquad \leq  \frac{1 - \Phi\left( \frac{\eta}{\sqrt{(r-1)}L_r} \right) }{ \Phi\left( \frac{c}{\sqrt{(r-1)} L_r} \right)}.
\end{align*}
\end{lemma}
\begin{proof}
Note that $\hat{h}_t^u$ is the original count $h_t^u$ plus Gaussian noise truncated at the threshold $m_\delta$.  More specifically, let $\cI_r(t)$ be the set of indices in the partial histogram table that get used to compute the counts at time $t$.  We then have $\hat{h}_t^u$ is distributed as a truncated (at $m_\delta$) Gaussian with mean $h_t^u$ and variance $|\cI_r(t)| L_r \tau^2$. Using the fact that $|\cI_r(t)| \leq (r-1)L_r$ we have
\begin{align*}
\Pr[ |h_t^u - \hat{h}_t^u| \geq \eta \tau \mid u \in \cD_t] &= \Pr_{Z \sim \Normal{h_t^u}{\cI_{r}(t) L_r \tau^2}}\left[ |Z - h_t^u| \geq \eta\tau \mid Z > m_\delta\right] \\
& \leq \Pr[Z < h_t^u - \eta\tau \mid Z > m_\delta] + \Pr[  Z > h_t^u + \eta\tau \mid Z > m_\delta] \\
& = \frac{2\left(1 - \Phi\left( \frac{\eta}{\sqrt{\cI_r(t) L_r}} \right) \right) - \Phi\left( \frac{-c}{\sqrt{\cI_r(t) L_r}} \right)}{1 - \Phi\left( \frac{-c}{\sqrt{\cI_r(t) L_r}} \right)} \\
& =  \frac{1 - 2\Phi\left( \frac{\eta}{\sqrt{\cI_r(t) L_r}}  \right) + \Phi\left( \frac{c}{\sqrt{\cI_r(t) L_r}} \right)}{\Phi\left( \frac{c}{\sqrt{\cI_r(t) L_r}} \right)} \\
& \leq \frac{1 - \Phi\left( \frac{\eta}{\sqrt{\cI_r(t) L_r}}  \right) }{\Phi\left( \frac{c}{\sqrt{\cI_r(t) L_r}} \right)} \\
& \leq  \frac{1 - \Phi\left( \frac{\eta}{\sqrt{r-1} L_r}  \right) }{\Phi\left( \frac{c}{\sqrt{r-1} L_r} \right)}
\end{align*}
\end{proof}

%% file: conclusion.tex
\section{Conclusion \label{sect:conclusion}} 
We have revisited the problem of releasing differentially private histograms in the continual observation model, introduced by \cite{DworkNaPiRo10} and \cite{ChanShSo11}.  We considered event level privacy, where events in a stream can consist of multiple elements, such as a purchase from a pharmacy would be an event yet a customer can purchase multiple drugs at that single event.  We then considered the various DP algorithms for the restricted/unrestricted $\ell_0$-sensitivity with known/unknown domain settings.  These various settings of releasing privatized histograms was originally introduced in \cite{DurfeeRo19} and \cite{RogersSuPeDuLeKaSaAh20}, but not for continual release.  We showed that we can use these existing DP algorithms for continual observation, but it required running the DP algorithms on various subsequences of the event streams, which might be prohibitively expensive in run time for many applications.  We then presented more practical DP algorithms that take the aggregated counts at each round to return a noisy histogram continually for the unrestricted $\ell_0$-sensitivity with unknown domain setting along with the unrestricted $\ell_0$-sensitivity with known domain setting.  There are multiple open research directions here, such as providing practical DP algorithms for the unrestricted $\ell_0$-sensitivity with unknown domain setting.  Further, are there optimal ways to set the thresholds in $\sparseGumb{s,k}$ and what are the various utility results with changing $s$ for each when compared to running $\knownBase{}$, which will have noise that depends on $d$?

%% file: appendixOldAlgos.tex
\section{Algorithms from Table~\ref{table:tasks} \label{appendix:oldAlgos}}

\begin{algorithm}[h!]
	\caption{$\knownGauss{}$; Gaussian mechanism over known domain $\cU$ with $\ell_0$ sensitivity $\Delta_0$} \label{algo:kGauss}
	\begin{algorithmic} 
		\State \textbf{Input:} Histogram $\bbh = \{ (h_u,u) : u \in \cU \}$, along with noise scale $\tau$.
		\State \textbf{Output:} Noisy result.
            	\State Return $\{(\Normal{h_u}{\tau^2}, u): u \in \cU \}$
	\end{algorithmic}
\end{algorithm}

\begin{theorem}[\citet{BunSt16}] \label{lem:kGaussPriv}
Assume that $||\bbh - \bbh'||_\infty \leq 1$ and $|| \bbh - \bbh'||_0 \leq \Delta_0$ for any neighbors $\bbh,\bbh'$. Then, the algorithm $\knownGauss{}(\cdot; \tau)$ is $\tfrac{\Delta_0}{2\tau^2}$-zCDP and hence $\left(\tfrac{\Delta_0}{2\tau^2} + \tfrac{1}{\tau} \sqrt{2\Delta_0 \ln(1/\delta')}, \delta' \right)$-DP for any $\delta' > 0$. 
\end{theorem}

\begin{algorithm}[h!]
	\caption{$\knownGumb{k }$; Exponential Mechanism over known domain $\cU$} \label{algo:kGumb}
	\begin{algorithmic} 
		\State \textbf{Input:} Histogram $\bbh = \{ (h_u,u) : u \in \cU \}$, number of outcomes $k$, and noise scale $\tau$.
		\State \textbf{Output:} Ordered set of $k$ indices and counts.
		\State Set $S = \emptyset$
		\For{$u \in \cU$}
			\State Set $v_u = h_u + \gum(\tau/2)$
			\State Update $S \gets S \cup \{ (v_u,u) \}$
		\EndFor
		\State Sort $S$ where $v_{u_{1}} \geq \cdots \geq v_{u_{\vert \cU \vert}}$
		\State Return $\left\{ (\Normal{h_{u_{1}}}{ \tau^2}, u_{1} ),....,(\Normal{h_{u_{k}}}{\tau^2}, u_{k} )\right\}$	
	\end{algorithmic}
\end{algorithm}

\begin{theorem}[\citet{CesarRo20}]
Assume that $||\bbh - \bbh'||_\infty \leq 1$ and $|| \bbh - \bbh'||_0$ is unrestricted for any neighbors $\bbh,\bbh'$. Then, $\knownGumb{k }(\cdot; \tau)$ is $\tfrac{k}{\tau^2}$-zCDP and hence $\left(\tfrac{k}{\tau^2} + \tfrac{2}{\tau} \sqrt{k \ln(1/\delta')}, \delta' \right)$-DP for any $\delta' > 0$. 
\end{theorem}

\begin{algorithm}[h!]
	\caption{$\unkGumb{k,\bar{k} }$; Unknown domain mechanism with access to $\bar{k}+1 \geq k$ elements} \label{algo:uGumb}
	\begin{algorithmic}
		\State \textbf{Input:} Histogram $\bbh$; outcomes $k$, cut off at $\bk+1$, and noise scale $\tau$.
		\State \textbf{Output:} Ordered set of indices and counts.
		\State Sort $h_{(1)} \geq h_{(2)} \geq \cdots \geq h_{(\bar{k} + 1)}$, with respective labels $i_{(1)}, \cdots, i_{(\bk+1)}$.
		\State Set $h_{\bot} = h_{(\bk+1)} + 1 + \tau \ln(1/\delta)$, with label $\bot$.
		\State Set $v_{\bot} = h_{\bot} +  \gum(\tau/2)$, with label $\bot$.
		\State Set $S = \emptyset$
		\For {$j \leq \bk$}
				\If{ $h_{(j)} > h_{(\bar{k} + 1)}$ }
        					\State Add $S \gets S \cup \{ (v_{(j)} = h_{(j)} + \gum(\tau/2), i_{(j)}) \}$.
				\EndIf
		\EndFor
        		\State Sort $S$ by its counts.
        		\State Let $v_{1 },....,v_{j},v_{\bot}$ be the descending list of counts up until $v_{\bot}$, with respective labels $u_1, \cdots, u_j$.
		\If{$j < k$}
			\State Return $\{( \Normal{h_{u_{1}}}{\tau^2}, u_{1} ),...,(  \Normal{h_{u_{j}}}{ \tau^2}, u_{j} ),\bot\}$
		\Else
			\State Return $\{(\Normal{h_{u_{1}}}{\tau^2}, u_{1} ),...,( \Normal{h_{u_{k}}}{\tau^2}, u_{k} )\}$
		\EndIf
	\end{algorithmic}
\end{algorithm}

\begin{theorem}[\citet{DurfeeRo19}] \label{thm:UnkGumbelPrivacy}
Assume $||\bbh - \bbh'||_\infty \leq 1$ for any neighbors $\bbh,\bbh'$. For any $\delta > 0$, $\unkGumb{k,\bar{k} }(\cdot; \tau,\delta)$ is $\left(\frac{k}{\tau^2} + \frac{2}{\tau} \sqrt{k \ln\left( 1/\delta'\right)},\bk\delta + \delta'\right)$-DP.
\end{theorem}